%% file: main.tex
\documentclass{article}
\usepackage{mystyle}

\usepackage{braket}

\newcommand{\starp}{\star} 
\newcommand{\id}{\mathrm{id}}

\title{Deformation quantisation of the conic and symplectic reduction of wavefunctions}

\author{Michael Swaddle \\
    \tt{\href{mailto:meswaddle@proton.me}{meswaddle@proton.me}
    }
}

\date{\today}
\begin{document}

    \maketitle
    \begin{abstract}
        \noindent We give a short review of the algebraic procedure known as deformation quantisation, which replaces a commutative algebra with a non-commutative algebra. We use this framework to examine how the objects known as wavefunctions, as known in the quantum curve literature, arise from deformation quantisation. We give an example in terms of the planar conic \(-y+x^2+2 xy + y^2\), and construct an associated wavefunction. We also give an example of the symplectic reduction of a wavefunction, following a procedure from Kontsevich and Soibelman \cite{ks_airy}.
    \end{abstract}

    \tableofcontents 
    
    \input{intro.tex}

\input{def_quant.tex}
    \input{conic} 
    \input{reduction}
    \input{ack}
    

    \printbibliography
    
\end{document}

%% file: intro.tex
    \section{Introduction}

    Broadly, \emph{quantisation} is an algorithm that embeds a commutative algebra inside a non-commutative algebra, preserving some structure. One approach to quantisation is \emph{deformation quantisation} \cite{chen_thesis,gunningham, k_defquant, b_defquant, quant_for_phys, k_defofPois}, motivated by quantisation in physics. Let \( \mathbf{k}\) be a field of characteristic zero.
    Deformation quantisation seeks to replace a commutative \( \mathbf{k}\)-algebra \(A\) with a non commutative \( \mathbf{k}\lBrack \hslash \rBrack\)-algebra \( \mathcal{A}_{\hslash}\), such that under the quotient map \( \mathcal{A}_{\hslash} / ( \hslash \mathcal{A}_{\hslash}) \simeq A\). Note that deformation quantisation is generally not a quantisation rule for sections, \(a \in A \rightarrow a_\hslash \in \mathcal{A}_{\hslash}\), instead it is a global operation replacing an entire algebra or sheaf of algebras. 

    
    Deformation quantisation can be applied to many different objects, for example the deformation quantisation of functions on a Poisson manifold by Kontsevich \cite{k_defofPois}. More generally for algebraic varieties and schemes one can consider the deformation quantisation of sheaves, such as in \cite{yekutieli,k_defquant}. It is even possible to consider fields with positive characteristic \cite{k_holonomic}.

    In this article we first give a brief overview of deformation quantisation. Next we discuss the sheaf theoretic interpretation of a \emph{wavefunction}. Not to be directly confused with the term as used in physics, wavefunctions are generators of modules over a non-commutative algebra analogous to \( \mathcal{A}_{\hslash}\). These wavefunctions are studied in the context of \emph{topological recursion} and \emph{quantum curves}, where they contain interesting enumerative data \cite{norbury_quant}.
    
    In the context of quantum curves, a wavefunction \( \psi(x) \), is a WKB solution to a differential equation of the form \( \widehat{H} \psi(x) = 0\), where \(\widehat{H}\) is an operator. The WKB method means looking for a series of the form \( \psi(x) = \exp( \sum \hslash^g S_g(x) )\), where \(\hslash\) is a formal parameter. One question commonly asked, for example in \cite{norbury_quant, tudor,abpolyquant}, is if there is a canonical or natural way to relate an operator \( \widehat{H}'(x , \hslash \partial) \) and a planar curve \(H(x,y) \in \mathbf{k}[x,y]\). In this paper, we try to motivate, with an example, a sheaf theoretic explanation arising from deformation quantisation. From this perspective, the answer is no, there are multiple choices of \( \widehat{H}\) which in the classical limit (a quotient by \(\hslash^2\)) produces a curve \( H(x,y)\). Deformation quantisation embeds a commutative algebra \(A\) within an algebra of operators \( \mathcal{A}_\hslash \), so there is not a natural isomorphism between \( A\)-modules and \( \mathcal{A}_{\hslash}\)-modules. 
    
    Given a variety \(X= \mathrm{Spec}(A)\), we consider when there exists a deformation quantisation of \(A\), given by \(\mathcal{A}_{\hslash}\). Inside \(X\), consider a Lagrangian subvariety \( \mathbb{L} \subset X\), corresponding to the quotient \(E=A/I\). The quantisation of \( \mathbb{L}\) is represented by a \(\mathcal{A}_{\hslash}\)-module, \(E_\hslash\) which under a quotient map by \(\hslash\) recovers \(E=A/I\). The algebra \( E_{\hslash}\) can be thought of as operators on \( \mathbb{L}\).
    Correspondingly a wavefunction is an element of the dual module \(E_\hslash^*\), which is the space of maps from \(E_\hslash\) to  \(\mathcal{A}_{\hslash} \). Intuitively, the dual of an operator is something like a function, they can pair to produce another function.
    
    This is mostly done by taking the equations of the defining ideal \(I\), and then looking at the the quantisation \(I_{\hslash}\), which is represented by differential operators on \(X\). A wavefunction is also seen as a solution, or an annihilator, to these differential equations.
    
    
    So the purpose of this article is to highlight that wavefunctions arise from deformation quantisation. As an example we construct a wavefunction for the planar conic in section (\ref{sec:def_of_conic}), which gives rise to the \emph{Airy equation} \cite{airy}. Finally we examine a natural transformation of wavefunctions associated with \emph{symplectic reduction}. This transformation, observed in \cite{ks_airy}, behaves like a Laplace or Fourier transform, and we give an example in the case of a four dimensional variety in section (\ref{section:wavefunction_reduction}).

    

%% file: def_quant.tex
    \section{Non-commutative formal deformations of the structure sheaf}
    Let \((X,\mathcal{O}_X)\) be a scheme over \( \Spec( \mathbf{k})\), where \( \mathbf{k}\) is a field of characteristic zero.  

    \begin{defn}[Formal deformation]
    A formal deformation of \((X,\mathcal{O}_X)\) is a sheaf \( \mathcal{A}_{\hslash} \) of \( \mathbf{k}\lBrack \hslash \rBrack\)-flat, associative \( \mathbf{k}\lBrack \hslash \rBrack \)-algebras on \(X\) with the constraint that \begin{equation} 
    \label{eqn:def_cons}
    0 \rightarrow \hslash \mathcal{A}_\hslash \rightarrow  \mathcal{A}_{\hslash} \rightarrow \mathcal{O}_X  \rightarrow 0,
    \end{equation}
    is exact.
    \end{defn}
    In particular, for deformation quantisation, we consider when these are \emph{non-commutative} \( \mathbf{k}\lBrack \hslash \rBrack\)-algebras.
    
    Alternatively, the formal deformation of \(\mathcal{O}_X\) is a sheaf \( \mathcal{A}_{\hslash}\), which fits into the pullback square, induced from the sequence (\ref{eqn:def_cons}):
    \begin{center}
        \begin{tikzcd}                  \mathcal{O}_X\otimes_{\mathbf{k}\lBrack \hslash \rBrack } \mathcal{A}_{\hslash} \cong \mathcal{O}_X  &  \arrow[l] \mathcal{A}_{\hslash} \\
                \mathbf{k} \arrow[u]  &  \arrow[l] \mathbf{k}\lBrack \hslash \rBrack \arrow[u] 
        \end{tikzcd}
    \end{center}
    \begin{rem} 
    Deformations are often seen in commutative contexts. For example in \cite{hartshorne_def}, a deformation refers to a family of commutative algebras giving rise to schemes. In particular \cite{hartshorne_def}, a deformation of a scheme \(X_0\), is a flat map of schemes \(\mathcal{X} \rightarrow S\) such that there is distinguished point in \(S\), so \(X_0\) fits into the following pullback diagram:
    \begin{center}
        \begin{tikzcd}
            X_0 \arrow[d] \arrow[r] & \arrow[d] \mathcal{X}\\ 
            \{ \text{pt} \} \arrow[r] &  S
        \end{tikzcd}
    \end{center}
    The similarity between the deformations in \cite{hartshorne_def} and the deformations appearing here in deformation quantisation is the flatness of \( \mathcal{A}_{\hslash} \) as a \( \mathbf{k}\lBrack \hslash \rBrack\)-module, and the pullback square, even though there is not necessarily a spectrum or space associated to \( \mathcal{A}_{\hslash}\). The other difference is that there is no requirement for commutativity. Deformations in general usually only require some associative structure \cite{k_def_book}. Another example of a non-commutative deformation could be \emph{super-schemes} \cite{super}, which are \(\mathbb{Z}_2\)-graded algebras. There is a similar notion of a pullback square with a commutative algebra.
    \end{rem} 
    
    A formal deformation of \( (X,\mathcal{O}_X)\) is a limit \(\mathcal{A}_{\hslash}\) given by successive non-commutative \( \mathbf{k}[\hslash]/\hslash^n\)-deformations, to give \( \mathcal{A}_{\hslash}\), a \( \mathbf{k} \lBrack \hslash \rBrack \)- deformation. 
    \begin{center} 
    \begin{tikzcd}
        \mathcal{O}_X  & \arrow[l]   \mathcal{A}_1 & \arrow[l] \dots  & \arrow[l]  \mathcal{A}_{\hslash}  \\ 
        k \arrow[u] &  \arrow[l] \arrow[u]  \mathbf{k}[\hslash]/\hslash^2  &  \arrow[l] \dots &  \arrow[l] \arrow[u] \mathbf{k} \lBrack \hslash \rBrack  
    \end{tikzcd}
    \end{center} 
    


   
    \begin{rem} The isomorphism 
    \( \mathcal{A}_\hslash \cong \mathcal{O}_X \otimes_{\mathbf{k}\lBrack \hslash \rBrack} \hslash \mathcal{A}_\hslash\) in diagram (\ref{eqn:def_cons}) is called an \textit{augmentation}.
    \end{rem}

    \begin{defn}[Star-product]
    \label{defn:star_prod}
    Let \(f,g \in  \mathcal{A}_{\hslash}\) be sections. A \emph{star product}, \cite{collini}, is a \(\mathbf{k}\lBrack \hslash \rBrack \)-bilinear, associative, map
    \(\star : \mathcal{A}_{\hslash}\otimes_{\mathbf{k}\lBrack\hslash\rBrack} \mathcal{A}_{\hslash} \rightarrow \mathcal{A}_{\hslash}\)
    given by the formal sum:
    \begin{equation} 
        \label{eqn:star_prod_formal}
        f \star g = \sum_k \omega_k(f,g) \hslash^k
    \end{equation}
    such that \(\omega_0 = \mathrm{prod} \), and \(\omega_k\) is a \( \mathbf{k}\lBrack \hslash \rBrack\)-bilinear map. 
    \end{defn}

    \begin{corollary}
    A choice of star product gives a formal non-commutative deformation of \( \mathcal{O}_X\). 
    \end{corollary}
    In particular, for deformation quantisation, this star-product will be required to satisfy extra constraints. 
    
    The requirement that \( \star \) is associative, recursively places constraints on \( \omega_k\). Associativity of \( \star\) means that:
    \begin{align*}
        (f \star g) \star h = f \star (g \star h).
    \end{align*}
    Substituting in the definition of \( \star\) as a formal power series, and using the initial condition \( \omega_{0}(f,g) = fg\), shows that \( \omega_k\), for \( k \geq 1\) must satisfy the recursive formula: 
    \begin{align}
        \sum_{j,k} \hslash^{j+k} \omega_j (\omega_k ( f ,g ),h) = \sum_{j,k} \hslash^{j+k} \omega_k(f,\omega_j(g,h)).
    \end{align}
    For example, looking at order \( \mathcal{O}(\hslash)\) terms gives the condition that \( \omega_1\) must satisfy:
    \begin{align*}
    \omega_{1}(f g, h)  +  \omega_1(f,g) h  - \omega_{1}(f, g h)  - f \omega_{1}(g,h) = 0.
    \end{align*}    
    \begin{rem} Note that in general, the existence of a star product is not guaranteed. At some point, solving for \( \omega_k\) can fail.
    \end{rem}
    
    If we denote \( \star^{g}\) as a star product truncated to powers of \(\hslash^{g-1}\), a formal deformation is given by a limit 
    \( \lim_g (\mathcal{O}_X \otimes_{\mathbf{k}} \mathbf{k}[\hslash]/\hslash^g, \star^g) = (\mathcal{O}_X \lBrack \hslash \rBrack, \star ) \)


    \subsection{Deformations of sheaves of \texorpdfstring{\(\mathcal{O}_X\)}{O\_X}-modules } 
    
    Let \((X,\mathcal{O}_X)\) be a scheme over \( \Spec(\mathbf{k})\). Let \( \mathcal{F}\) be a quasicoherent sheaf of \( \mathcal{O}_X\)-modules, and so consider \( \mathcal{F}\) as a sheaf of \( \mathbf{k}\)-modules.
    
    Consider a formal deformation \( \mathcal{A}_{\hslash}\) of \(\mathcal{O}_X\). The formal deformation of \( \mathcal{F}\) as a \( \mathbf{k}\)-module is given by:
    \begin{defn}[Formal deformation of a sheaf]
    A \emph{formal deformation} of a sheaf \( \mathcal{F}\) of \( k\)-modules is a sheaf \( \mathcal{F}_{\hslash} \) of \( \mathbf{k}\lBrack\hslash\rBrack\)-modules, such that \( \mathcal{F}_{\hslash}\) 
    is \(\mathbf{k}\lBrack\hslash\rBrack \)-flat and satisfies: 
    \[ 0 \rightarrow \hslash \mathcal{F}_{\hslash} \rightarrow  \mathcal{F}_{\hslash} \rightarrow \mathcal{F}  \rightarrow 0.\]
    \end{defn}
    
    There are conditions on the existence of non trivial deformations. Theorem (\ref{thm:torsionfree}) is an obstruction to the existence of a particular formal deformation in the case of a Poisson scheme. Further note that the flatness condition rules out \( \mathcal{O}_X\) and similarly \( \mathcal{F}\) as a deformation, as \( \mathcal{O}_X\) or \( \mathcal{F}\), is not flat as a \( \mathbf{k}\lBrack \hslash\rBrack\)-algebra.

    There is the possibility of multiple deformations, Let \( \mathcal{E}_\hslash \) and \( \mathcal{F}_\hslash\) be sheaves of  \(  \mathbf{k}\lBrack\hslash\rBrack\)-modules. It can be the case that \( \mathcal{E}_\hslash \ncong \mathcal{F}_\hslash\), but \( \mathcal{E}_\hslash/ \hslash \mathcal{E}_{\hslash} \cong \mathcal{F}_\hslash/ \hslash \mathcal{F}_{\hslash} = E\), where \( E \) is a sheaf of \( \mathcal{O}_X\)-modules.

    
 
    It is necessary to consider formal deformations in a neighbourhood of \(\hslash = 0\). Although sometimes written as setting \(\hslash\rightarrow 0\), this is not a well defined operation.
    
    \begin{defn}[Localisation] Define \( \mathcal{F}_{\hslash}^{\mathrm{loc}} \cong \mathcal{F}_{\hslash} \otimes_{\mathbf{k}\lBrack \hslash \rBrack} \mathbf{k} \lParen \hslash \rParen\), as the localisation at \( \hslash\).
    \end{defn}
    
    \begin{ex}Consider \( A_{\hslash} = \mathbf{k} \lBrack \hslash \rBrack\). Then
    \(A^{\mathrm{loc}}_{\hslash} = \mathbf{k}\lBrack \hslash \rBrack [  \hslash^{-1} ] \cong \mathbf{k} \lBrack \hslash \rBrack  [\lambda ] / ( \hslash \lambda - 1) \cong \mathbf{k} \lParen \hslash \rParen \).
    \end{ex}    
    
    It will be necessary to allow \(\hslash\) to be invertible. Although in certain examples, deformation quantisation will produce a space of operators, there is not necessarily a corresponding space of solutions, unless \(\hslash\) is allowed to be invertible or non-zero. We will show this in the example of the conic, example (\ref{ex:the_conic}). This naively makes sense, as quantum objects do not necessarily need to have a classical limit.

    \begin{ex}
    \label{ex:poisson_brack}
    The star product defines a Poisson bracket under restriction to \( \mathcal{O}_X\).
    \( \frac{1}{\hslash} \left( f \star g - g \star f\right) \) restricted to \( \mathcal{O}_X\) is a Poisson bracket.
    \end{ex}

    \emph{Deformation quantisation} of a Poisson scheme/variety is the inverse problem to example \ref{ex:poisson_brack}. Starting with the data of a Poisson scheme/variety, deformation quantisation is a formal deformation such that the star product recovers the Poisson bracket mod \( \hslash\). An example  (\ref{ex:moyal_prod}) of such a star product is called the \emph{Moyal} product.

    \section{Deformation quantisation}
   
    
    Starting with the data of a Poisson scheme or variety, what is termed \emph{deformation quantisation} in the literature, for example \cite{yekutieli}, is a formal non-commutative deformation with a choice of star product which recovers the Poisson bracket mod \( \hslash\).  Deformation quantisation is the inverse problem to example (\ref{ex:poisson_brack}). Deformation quantisation is a choice of star product that agrees with the Poisson structure. Note that deformation quantisation occurs at the level of functions, the meaning of some space associated to these non-commutative objects is not clear.
    
    Note in the following sections, for a module \(M_0\), we use the notation that
    \( M_0 \lBrack \hslash \rBrack = \lim_n (M_0 \otimes_{\mathbf{k}} \mathbf{k} \lBrack \hslash \rBrack)  / ( 1 \otimes_{\mathbf{k}} \hslash)^n\).

    \subsection{Poisson bi-vectors and bi-maps}
    
    Let \( (X,\mathcal{O}_X)\) be a Poisson scheme over a field \( \mathbf{k}\), which means it is equipped with a \emph{Poisson bracket}, a skew-symmetric bilinear map on sections \( \{ \cdot ,\cdot \} : \mathcal{O}_X \otimes_{\mathbf{k}} \mathcal{O}_X \rightarrow \mathcal{O}_X \), satisfying the Leibniz rule and the Jacobi identity. The data of a Poisson bracket is captured in an object called a \emph{Poisson bi-vector}:

    \begin{defn}[Poisson bi-vector]
    \label{defn:bi_vect}
    A \emph{Poisson bi-vector} is a choice \( \Pi \in   H^0(X,\bigwedge^2 \mathcal{T}_X) \), such that for local sections \( f ,g : \mathcal{O}_X\), the Poisson bracket is computed via
    \begin{equation}
    \label{eqn:bi_vect}    
    \{ f,g\} = (df \otimes_{\mathbf{k}} dg ) ( \Pi) = \sum_{i<j} \pi^{ij} \partial_i f \partial_j g.
    \end{equation}
    where \( \Pi =  \frac{1}{2} \pi^{ij} \frac{\partial}{\partial x_i} \wedge \frac{\partial}{\partial x_j} \), and \( \pi^{ij}\) is skew-symmetric.
    \end{defn}
    
    \begin{ex} 
    A bi-vector identified with derivations acts as  map \( \mathcal{O}_X \otimes_{\mathbf{k}} \mathcal{O}_X \rightarrow \mathcal{O}_X\). It takes a pair of functions and produces a function, via the formula 
    \[ (X \wedge Y) ( f \otimes_{\mathbf{k}} g) = X (f) Y (g) - Y (f) X (g). \]
    Alternatively a bi-vector can be evaluated by a bi-differential to produce a function.
    
    Consider a two dimensional example with canonical coordinates \( (x_1,x_2)\), so \( \pi_{12}=1\), \(\pi_{21} = -\pi_{12}\), \(\pi_{11}=\pi_{22}=0\), and 
    \[ \Pi = \frac{1}{2} \pi^{ij} \frac{\partial}{\partial x_i }\wedge \frac{\partial}{\partial x_j } =  \frac{\partial}{\partial x_1 }\wedge \frac{\partial}{\partial x_2 }.\]
    Taking the differential 
    \( ( d f \otimes d g )\) and evaluating on \( \Pi \) gives
    \[ ( d f \otimes d g )( \Pi ) =  \frac{ \partial f}{\partial x_1}\frac{\partial g}{\partial x_2} - \frac{\partial f}{\partial x_2 }\frac{\partial g}{\partial x_1}.\]
    This matches the application
    \[ \frac{\partial}{\partial x_1 }\wedge \frac{\partial}{\partial x_2 } ( f \otimes g) = \frac{ \partial f}{\partial x_1}\frac{\partial g}{\partial x_2} - \frac{\partial f}{\partial x_2 }\frac{\partial g}{\partial x_1}. \]

    \end{ex}

    Closely related to the bi-vector \( \Pi\), is an endomorphism  \( \pi\), such that \( \pi : \mathcal{O}_X \otimes_{\mathbf{k}} \mathcal{O}_X \rightarrow \mathcal{O}_X \otimes_{\mathbf{k}} \mathcal{O}_X\).
    \( \pi \) is built from the coefficients of \( \Pi\), as follows:  
    \[ \pi(f\otimes_{\mathbf{k}} g) = \pi^{ij} \partial_i (f) \otimes_{\mathbf{k}} \partial_j (g).\]
    
    The difference between the two is an application of a natural linear map \( \mathrm{prod} : \mathcal{O}_X \otimes_{\mathbf{k}} \mathcal{O}_X \rightarrow \mathcal{O}_X\), which is the pointwise product of functions: 
    \[ \mathrm{prod}(f \otimes_{\mathbf{k}} g ) = f g.\]
    Composing \(\pi\) with \( \mathrm{prod} \) gives the Poisson bracket:
    \[ \mathrm{prod} \circ \pi (f\otimes_{\mathbf{k}} g) = \{ f, g\}. \]
    Note that \( \mathrm{prod} \circ  \pi \) satisfies the property, \( \mathrm{prod} \circ  \pi ( f \otimes_{\mathbf{k}} g) = -\mathrm{prod} \circ  \pi ( g \otimes_{\mathbf{k}} f ) \), by the skew-symmetry of the \( \pi^{ij}\).

    \begin{defn}[Poisson endomorphism or bi-map]
    Define the map \( \pi : \mathcal{O}_X \otimes_{\mathbf{k}} \mathcal{O}_X \rightarrow \mathcal{O}_X \otimes_{\mathbf{k}} \mathcal{O}_X\) to satisfy:  
    \[  \mathrm{prod} \circ \pi (f,g) = \{ f, g\}. \]
    \end{defn}
    
    
    \subsubsection{Formal Poisson structures}
    \begin{defn}[Formal Poisson structure] 
    A \emph{formal Poisson structure}, is any extension of a \hyperref[defn:bi_vect]{bi-vector}, equation (\ref{eqn:bi_vect}), by a formal power series:
    \[ \Pi_{\hslash} =  \sum_{g=1}^{\infty} \Pi_g \hslash^{g-1} \in H^0(X,\bigwedge\nolimits^{\!2} \mathcal{T}_X)\lBrack \hslash \rBrack , \]
    where \(\Pi_1 = \Pi\), such that 
    \( \Pi_{\hslash} \) is a Poisson structure on \( \mathcal{O}_X \lBrack \hslash \rBrack \).
    \end{defn}

    There are corresponding endomorphisms or bi-maps associated with \( \Pi_g\) like the Poisson case. 
    
    There is the question of what formal Poisson structures correspond to a formal deformation of a Poisson scheme. For example, does a formal Poisson structure correspond to a star product. The next section we give one particular answer to this, in terms of a special star product constructed as an exponential.
    
    \subsection{A deformation quantisation of Poisson schemes}
    \label{sec:def_of_pois_sch} 
    
    A deformation quantisation of a Poisson scheme \( \mathcal{O}_X\), is a sheaf \( \mathcal{A}_{\hslash} = \mathcal{O}_X \lBrack \hslash \rBrack \) equipped with a star product \cite{cttaneo_star, k_defofPois, groenewold} that recovers the data of the Poisson bi-vector under the quotient by \(\hslash\), and gives \( \mathcal{O}_X \lBrack \hslash\rBrack\) a formal Poisson structure. Note that this is not yet fully analogous to a \emph{quantisation} as used in physics, as there is no Hilbert space. \( \mathcal{A}_{\hslash}= \mathcal{O}_X \lBrack \hslash \rBrack\) is analogous to a phase space representation of operators in quantum mechanics. Looking for a Hilbert space of functions on which these operators act, or more specifically an \( \mathcal{A}_\hslash\)-module is a more complete analogy. 
    
    Given a Poisson bracket or bi-vector \(\Pi\), and an associated bi-map \( \pi\), a star product \( \starp_{\pi}\), is constructed by requiring, that:
    \[ [f,g] = f \starp_{\pi } g - g \starp_{\pi} \,f = \hslash \{f,g\} + \hslash^2 \,\Pi_{1}(f,g) + O(\hslash^3),   \]
    and the higher order terms correspond to a formal Poisson structure \( \Pi_{\hslash}\).
    
    Higher order terms are chosen or constructed recursively to satisfy associativity. \( \starp_{\pi}\) contains half the information of the formal Poisson structure. One solution to this problem, when the formal Poisson structure is simply iterated applications of \( \Pi_g = \mathrm{prod} \cdot \pi^g \), is the following star product:
    
    \begin{defn} 
    [Star-product associated to a Poisson bi-vector]\label{defn:star_prod_pois}
    Associated to \( \Pi\), define a star product \( \starp_{\pi}  : \mathcal{A}_\hslash \otimes_{\mathbf{k}\lBrack \hslash \rBrack}  \mathcal{A}_\hslash \rightarrow \mathcal{A}_\hslash \) locally via the series:
    \begin{align}
    \label{eqn:star_prod}
    f \starp_{\pi} g  &= \mathrm{prod} \circ \exp \left( \frac{\hslash}{2} \,\pi \right) (f \otimes g) , \\
    f \starp_{\pi} g   &= f g + \frac{1}{2} \hslash \,\pi^{ij}\, \partial_i f \partial_j g + \frac{1}{4} \hslash^2 \pi^{ik} \pi^{jl} \, \partial_i \partial_j f \partial_k \partial_l g + O(\hslash^3) . 
    \end{align}
    \end{defn}
    Non commutativity is given by the skew-symmetry of \( \pi \). This star product is chosen or constructed because it gives a representation into a Weyl-algebra, as we see in example (\ref{ex:the_conic}). Let \( \mathcal{W}_{\hslash}\) be a Weyl-algebra \( \varphi : \mathcal{A}_{\hslash} \rightarrow \mathcal{W}_{\hslash}\). This star product is constructed to satisfy \( \varphi( f \star g) = \varphi(f) \cdot \varphi(g)\), where \(\cdot\) is a product of operators.

    
    
    \begin{rem}[Notation in definition (\ref{defn:star_prod_pois})]
    Note in this definition.  \( \exp \) is the formal sum \( \exp(s) = \id + s + \frac{1}{2} s^2 + \dots\). Here \(  \id \) is the identity map \( \id (f \otimes g) = f \otimes g\). \( \mathrm{prod}\) is pointwise function multiplication, \( \mathrm{prod}(f \otimes g ) = f g\). Also \( \cdot\) is function composition. 
    
    Powers of \( \pi \) refer to iterated application \[ \pi^k ( f \otimes g ) = \pi^{k-1} ( \pi(f \otimes g ) ). \]  
    The endomorphism \( \pi \) is needed because there is no natural way to compose bi-vectors \( \Pi\).
    \end{rem}

    The coefficient of \( \hslash\) in \( \starp_{\pi}\) is identified with a coefficient of \( \Pi_g\) in a formal Poisson structure. The star product is also found by fixing \( \omega_0 = \mathrm{prod} \cdot 
    \id \), \( \omega_1 = \mathrm{prod} \cdot \pi\), and  \( \omega_k = \mathrm{prod} \cdot \pi^k \) in the definition of the star product (\ref{eqn:star_prod_formal}). A good reference on expanding out the star product and solving for explicit coefficients is \cite{starprodeasy}.

    Modulo \(\hslash^2\), the star product \( \starp_{\pi}\) recovers the Poisson bracket when applied to sections of \( \mathcal{O}_X\) in \( \mathcal{A}^{\text{loc}}_{\hslash}\),  \( \eta : \mathcal{A}^{\text{loc}}_\hslash \rightarrow \mathcal{O}_X\),
    \( \eta ( f \starp_{\pi} g )= f g \). Then
    \begin{align}
        \label{eqn:star_to_pois}
        \frac{1}{\hslash } \left( f \starp_{\pi} g - g \starp_{\pi} f \right) = \{f,g\} + O(\hslash).
    \end{align}    
    %
    %
    %
    A famous example of such a star product is the \emph{Moyal product}, defined in the case of a linear Poisson manifold.
    
    \begin{ex}[Moyal product]
    \label{ex:moyal_prod} 
    Consider the smooth symplectic case \(X = T^{*}C\), \(\dim(X)=2n\) with local coordinates \(\{x_1,\dots,x_n\}\) on a fibre. 
    Let  \( \pi : \mathcal{O}(X) \otimes \mathcal{O}(X) \rightarrow \mathcal{O}(X)\otimes \mathcal{O}(X)\) be skew-symmetric and written in terms of local coordinates as:
    \[ \pi(f \otimes g) = \pi^{ij} \frac{\partial f }{\partial x_i} \otimes \frac{\partial g}{\partial x_j }.\] 
    In this case there is a star product called the \emph{Moyal product}, \( \starp_M\) on \( \mathcal{A}_{\hslash} = \mathcal{O}(X)\lBrack \hslash \rBrack\). So expanding formula (\ref{eqn:star_prod}) gives:
    \[ f \starp_M g = f \cdot g + \frac{1}{2} \hslash \, \pi^{ij} \frac{\partial f}{\partial x_i} \cdot  \frac{\partial g}{\partial x_j} + \frac{\hslash^2}{4} \pi^{ij} \pi^{kl}\frac{\partial f}{\partial x_i \partial x_k} \cdot \frac{\partial g}{\partial  x_j \partial x_l} + \mathcal{O}(\hslash^3),\]
    where repeated indices are summed over.
    \end{ex} 
    
    The Moyal product from example (\ref{ex:moyal_prod}), is representable as the formal expansion of a formal Gaussian integral:
    \begin{lem}[\cite{baker}]
    If \( \star_M \) is the Moyal product per example (\ref{ex:moyal_prod}), then:
    \[f \starp_{M} g(q) = \frac{1}{Z} \int_X D p\, Ds\, f(p) g(s) \, \exp\left( -\frac{1}{2 \hslash} \omega(p-q, s-q) \right),\]
    where \(\omega\) is the quadratic form given by \( \omega =  (\pi^{-1})\) , \( Z = \int D p Ds \, \exp\left( -\frac{1}{2 \hslash}  \omega(p,s)\right) \). 
    \end{lem}
    This result is originally due to \cite[p. 2199]{baker}, equation (3''). In finite dimensions, \(Z\) is given by a determinant. 
    \begin{proof}
    Integrating over \( x=(p,s)\) as a Gaussian integral shows the correspondence:
    \[  \frac{1}{Z} \int Dx \, \mathrm{prod} \circ  (f \otimes g) (x) \exp \left( - \frac{1}{2\hslash} Q(x-q \otimes q,x-q\otimes q) \right), \]
    where \(Q\) is represented as a block tensor, \(Q = \frac{1}{2}\left(\begin{array}{cc} 0 & \pi \\ \pi & 0 \end{array}\right) \), so  
    \[ = \mathrm{const} * \mathrm{prod} \circ \exp\left( \frac{1}{2\hslash} (Q)^{-1}_{ij} \partial_i \otimes \partial_j \right) (f\otimes g)(x)  \big|_{x \rightarrow (q,q) }.\] 
    \end{proof}

    So, the deformation quantisation of \( (X,\mathcal{O}_X)\) is a sheaf \( \mathcal{A}_{\hslash}\) of  \( \mathbf{k}\lBrack \hslash\rBrack\)-flat modules, equipped with a star product, such that the star product recovers the Poisson bracket via equation (\ref{eqn:star_to_pois}). 

    \begin{defn}[Deformation quantisation of structure sheaf] 
    \label{defn:def_quant}
    Let \( (X,\mathcal{O}_X)\) be Poisson with a bi-vector \( \Pi\). A \emph{deformation quantisation} of \( \mathcal{O}_X\) is given by the sheaf \( \mathcal{A}_{\hslash} = \mathcal{O}_X \lBrack \hslash \rBrack \), equipped with a star product \( \starp_{\pi} \).
    \end{defn}
    
    This result can be found in \cite{yekutieli}. Note that this does not necessarily correspond to quantisation from physics. Deformation quantisation has simply produced a non-commutative space of operators. In physics, quantisation also seeks to construct a Hilbert space of functions. Correspondingly we need to look for duals of \( \mathcal{A}_{\hslash}\)-modules, which is analogous to a Hilbert space. We give an example of the entire process in example (\ref{ex:the_conic}). The conditions on the existence of these extra modules is also an interesting question \cite[section 2.4, page 15]{abpolyquant}. 

    \emph{Flatness} is a technical condition, required to exclude the classical objects \( \mathcal{O}_X\) counting as a deformation quantisation of itself. For a definition of flatness see \cite{lang}. \(\mathcal{O}_X\) is not sheaf of flat \( \mathbf{k}\lBrack \hslash \rBrack\)-modules.

    \begin{lem} \( \mathcal{O}_X\) is not a sheaf of flat \( \mathbf{k}\lBrack \hslash \rBrack\)-modules.
    \end{lem}
    \begin{proof}
    Let \( I_{\hslash} = \langle \hslash \rangle \). 
    Consider the injective map of \( \mathbf{k} \lBrack \hslash \rBrack\)-modules:
    \[ 0 \rightarrow  I_{\hslash} \overset{\phi}{\rightarrow} \mathbf{k} \lBrack \hslash \rBrack,\]
    where \( \phi(f)  = f\), is inclusion. Now tensoring by \( \mathcal{O}_X\):
    \[ I_{\hslash} \otimes_{\mathbf{k}\lBrack \hslash \rBrack} \mathcal{O}_X \overset{\Phi}{\rightarrow} \mathbf{k}\lBrack \hslash \rBrack \otimes_{\mathbf{k}\lBrack \hslash \rBrack} \mathcal{O}_X. \]
    Now  \( \mathbf{k}\lBrack \hslash \rBrack \otimes_{\mathbf{k}\lBrack \hslash \rBrack} \mathcal{O}_X \cong \mathcal{O}_X\), so the induced map \( m \otimes f \rightarrow \phi(m) \otimes f \) cannot be injective.  
    \end{proof}

    Lemma (\ref{lem:flat_kh_mod}), describing flat \( \mathbf{k}\lBrack\hslash\rBrack\)-modules, requires an important result from commutative algebra:
    \begin{lem}[\cite{stacks_complete_flat}] 
    \label{lem:complete_flat}
    Let \(R\) be a Noetherian ring and \(I\) be an ideal. Then the completion of R with respect to I, \(\widehat{R}_I\) is a flat \(R\)-module.
    \end{lem}

    \begin{lem}
    \label{lem:flat_kh_mod}
    Let \(M_0\) be a \(\mathbf{k}\)-module. Then \( M_0 \lBrack \hslash \rBrack\) is a flat \( \mathbf{k}\lBrack \hslash\rBrack\)-module.
    \end{lem}

    \begin{proof}
    The completion \(M_0 \lBrack \hslash \rBrack\) is flat over \(M_0 \otimes_{\mathbf{k}} \mathbf{k}\lBrack \hslash \rBrack \) by lemma (\ref{lem:complete_flat}), where we use \(I = 1 \otimes \hslash\).  \(M_0 \lBrack \hslash\rBrack\) is the completion of \(M_0 \otimes_{\mathbf{k}} \mathbf{k}\lBrack \hslash \rBrack\) by \( 1 \otimes \hslash\): 
    \[ \lim_n \frac{M_0  \otimes_{\mathbf{k}} \mathbf{k}\lBrack \hslash \rBrack }{\left(1 \otimes_{\mathbf{k}} \hslash\right)^n} \cong \lim_n M_0 \otimes_{\mathbf{k}}  \frac{\mathbf{k}[\hslash]}{\hslash^n} \cong \lim_n \frac{M_0[\hslash]}{\hslash^n} \cong M_0 \lBrack\hslash\rBrack.\] 
    Then \(M_0 \otimes_{\mathbf{k}} \mathbf{k} \lBrack \hslash \rBrack \) is flat over \( \mathbf{k} \lBrack \hslash\rBrack\), so \(M_0 \lBrack \hslash \rBrack\) is flat over \( \mathbf{k} \lBrack \hslash \rBrack\).
    \end{proof}


    \subsubsection{Properties of the exponential star product}

    There is a natural gauge transform for Poisson and formal Poisson structures. 
    Let \( \starp_{\pi}\) be a star product associated to Poisson bi-vector \( \Pi\).
    Let \( \Pi \rightarrow \Pi + \Gamma \) be a Gauge transform of a Poisson bi-vectors, where \( \Gamma \in H^0(X,\wedge^2 \mathcal{T}_X)\), 
    \( \Gamma = \frac{1}{2} \gamma^{ij} \frac{\partial}{\partial x_i } \wedge \frac{\partial}{\partial x_j}\), and \( \gamma^{ij} = \gamma^{ji}\). Define a bi-map \( \gamma : \mathcal{O}_X \otimes \mathcal{O}_X \rightarrow \mathcal{O}_X \otimes \mathcal{O}_X\) associated to \( \Gamma \) as before. Then \( \Pi + \Gamma \) defines a corresponding star product \( \starp_{\pi + \gamma}\) as follows:
    
    \begin{lem}
    There is an isomorphism, \( (\mathcal{O}_X \lBrack \hslash \rBrack,
    \starp_{\pi} ) \cong ( \mathcal{O}_X \lBrack \hslash \rBrack, \starp_{\pi + \gamma} )\), given by a map 
    \(\psi : \mathcal{O}_X \lBrack \hslash \rBrack \rightarrow \mathcal{O}_X \lBrack \hslash \rBrack\)  
    \begin{equation} 
    \label{eqn:star_prod_iso}
    \psi(f) = \exp \left( \frac{\hslash}{4} \, \gamma^{jk} \partial_j \partial_k \right) f = \exp \left( \frac{\hslash}{2} \widehat{y} \right) f.
    \end{equation}
    where \( \widehat{y} = \frac{1}{2} \gamma^{jk} \partial_j \partial_k\).
    Correspondingly the star product is given by: 
    \[f \starp_{\pi + \gamma}g =  \mathrm{prod} \circ \exp  \left( \frac{\hslash}{2} ( \pi + \gamma ) \right) (f \otimes g) \]
    \end{lem}
    
    \begin{proof}
    We are required to show:
    \[ \exp \left( \frac{\hslash}{2} \widehat{y} \right) \circ  \mathrm{prod} \circ \exp  \left( \frac{\hslash}{2}  \pi   \right) = \mathrm{prod} \circ \exp  \left( \frac{\hslash}{2}  (\pi + \gamma ) \right) \circ \left( \exp \left( \frac{\hslash}{2} \widehat{y} \right) \otimes \exp \left( \frac{\hslash}{2} \widehat{y} \right) \right).  \]
    First note, via the product rule, 
    \begin{align*} \widehat{y} \circ \mathrm{prod} ( f \otimes g) &= \mathrm{prod} \circ  ( \widehat{y} f \otimes g + f \otimes \widehat{y} g + \frac{1}{2} \gamma^{jk} \partial_j f \otimes \partial_k g + \frac{1}{2} \gamma^{jk}  \partial_k f \otimes \partial_j g) \\ 
    &= \mathrm{prod} \circ  ( \widehat{y} f \otimes g + f \otimes \widehat{y} g + \gamma^{jk} \partial_j f \otimes \partial_k g ). 
    \end{align*}
    So 
    \begin{align*} \exp \left( \frac{\hslash}{2} \widehat{y} \right) \circ \mathrm{prod} \circ \exp \left( \frac{\hslash}{2} \pi  \right)  &= \mathrm{prod} \circ  \exp \left( \frac{\hslash}{2} ( \pi +  \id \otimes \widehat{y} + \widehat{y} \otimes \id +  \gamma ) \right) \\
    &= \mathrm{prod} \circ \exp  \left( \frac{\hslash}{2}  (\pi + \gamma ) \right) \circ \left( \exp \left( \frac{\hslash}{2} \widehat{y} \right) \otimes \exp \left( \frac{\hslash}{2} \widehat{y} \right) \right) 
    \end{align*}
    \end{proof}
    
    
    This extends more generally. Suppose for now \( (X, \mathcal{O}_X )\) is not necessarily Poisson. This exponential form of the star-product in general does not rely on a Poisson structure, only a bi-map, but there is always a skew-symmeterisation which may correspond to a Poisson structure. Consider the star product associated to any bi-map \( \tau : \mathcal{O}_X \otimes \mathcal{O}_X \rightarrow \mathcal{O}_X \otimes \mathcal{O}_X\) (\( \tau\) not necessarily skew-symmetric):
    \begin{defn}[Star product associated to a bi-map] Let \( \tau\) be non-symmetric.
    \[  f \starp_{\tau} g = \mathrm{prod} \circ \exp \left( \frac{\hslash}{2} \tau \right) ( f \otimes g) \]
   \end{defn} 
    The star product \( \starp_\tau \) is isomorphic to \( \star_{\tau_{\text{skew}}}\). With the skew-symmeterisation of \( \tau\), so \( \tau \rightarrow \tau_{\text{skew}} =\frac{1}{2} \left( \tau -  \tau^T \right) \), choose \( \gamma = \frac{1}{2}\left( \tau + \tau^T\right) \), then via the isomorphism in equation (\ref{eqn:star_prod_iso}),  \( \starp_{\tau} \rightarrow \starp_{\tau_{\text{skew}}}\).

    One extra observation, for skew-symmetric bi-maps, composition with a permutation \( P : \mathcal{O}_X \otimes \mathcal{O}_X \rightarrow \mathcal{O}_X \otimes \mathcal{O}_X \), where \(P(f \otimes g) = - g \otimes f\) also gives a Poisson structure.
    
    \begin{lem} Let \( \pi \) be skew-symmetric.
    Then \[  \widetilde{\pi} (f\otimes g)= \pi \circ P ( f \otimes g) = -\pi^{ij} \partial_i (g) \otimes \partial_j (f), \] also defines a Poisson bracket.
    
    \( \widetilde{\pi}\) also satisfies the Yang-Baxter equation:
    \[ (\id \otimes \widetilde{\pi}) \circ ( \widetilde{\pi} \otimes \id  ) \circ ( \id \otimes \widetilde{\pi}) =  ( \widetilde{\pi} \otimes \id ) \circ (\id \otimes \widetilde{\pi}) \circ  ( \widetilde{\pi} \otimes \id ), \] 
     applied to a section \( f \otimes g \otimes h  \in \mathcal{O}_X \otimes_{\mathbf{k}} \mathcal{O}_X \otimes_{\mathbf{k}} \mathcal{O}_X  \).
     
    Further with \( \widetilde{\Pi} = \mathrm{prod} \cdot \widetilde{\pi}\), \( \starp_{\widetilde{\pi}}\) also defines a star product equivalent to  \( \starp_{\pi}\).
    \end{lem}
    
    \begin{proof}
    Check with direct computation and use skew-symmetry, \( -\pi^{ij} = \pi^{ji}\), and then relabel the terms.
    \end{proof}
    
    So \( \pi^k\) terms in the star product are diagrammatically braiding two threads together \( k \) times.

    In summary, the star product \(\star_{\pi}\), gives a way to replace the standard commutative product of functions by a non commutative, associative product. 
    The purpose of the star product \( \star_{\pi}\) from definition (\ref{defn:star_prod_pois}), is that when \( \mathcal{O}_X\) is a ring of polynomial or formal polynomials functions on a completion, \((\mathcal{A}_{\hslash}, \star_{\pi})\), is representable by a Weyl-algebra or formal completion of a Weyl-algebra in section (\ref{sec:weyl_algebra}).

    \subsection{Coisotripic subschemes and an obstructions to existence}
    
    Let \((X,\mathcal{O}_X)\) be a Poisson scheme. Let the deformation quantisation of \(\mathcal{O}_X\) be given by \(\mathcal{A}_{\hslash} = \mathcal{O}_X\lBrack \hslash \rBrack\) as sheaf of \( \mathbf{k} \lBrack \hslash \rBrack\)-algebras. Let \( \mathcal{A}_{\hslash}\) be equipped with the star product per definition (\ref{defn:star_prod_pois}).  An important question is what obstructions are there for the existence of a sheaf \(\mathcal{F}_\hslash\) of \( \mathcal{A}_{\hslash}\) modules. 
    
    It is possible to find an obstruction, which has a geometric meaning, by considering \emph{torsion}. For a module \(M\) over a ring \(R\):
    \begin{defn}[Torsion element] An element \(m\) in \(M\) is a \emph{torsion element}, if there exists an \(r \) in \(R\) such that \(rm=0\).
    \end{defn}
    \(M\) is a torsion module if all its elements are torsion elements. \(M\) is \(S\)-torsion if the \(r\) form a subring \(S\). 
    
    Flatness implies torsion free, while the converse is only true in \emph{Pr\'ufer domains}. For a flat \(R\)-module \(M\):
    \begin{lem}
    If \(M\) is flat over \(R\), then \(M\) is torsion free over \(R\).
    \end{lem}
    
    \begin{proof}
    We sketch the commutative case first. This extends to non-commutative \(R\) by considering left and right modules and annihilators. The goal is to show \( \mathrm{Ann}_R(M) = 0\).

    Suppose for a contradiction \(M\) is flat, but there exists a non zero \(r\in \mathrm{Ann}_R(M)\). Consider the sequence \[ 0 \rightarrow R \rightarrow R \] 
    given by multiplication by \(r\). If \( M\) is flat then
    \[ 0 \rightarrow R \otimes M \rightarrow R \otimes M. \]
    But \(r\) is an annihilator, so this must be the zero map, and hence this is a contradiction.
    \end{proof}

    Let the deformation quantisation of \((X,\mathcal{O}_X)\) be given by \( \mathcal{A}_\hslash = ( \mathcal{O}_X \lBrack \hslash \rBrack, \starp_{\pi} )  \) as before. Suppose there exists a deformation quantisation of a sheaf \( \mathcal{F}\), given by \( \mathcal{F}_{\hslash}\).
    If \( \mathcal{F}_\hslash \) is flat over \( \mathbf{k} \lBrack \hslash \rBrack \), then this implies that \( \mathcal{F}_{\hslash}\) must be \( \hslash\)-torsion free. Then \(\hslash\)-torsion free gives the constraint, proven below in theorem (\ref{thm:torsionfree}), that the support of the sheaf \(\mathcal{F} \cong \mathcal{F}_{\hslash}/\hslash \mathcal{F}_{\hslash}\) must be \emph{coisotropic}. Recall coisotropic means \( \mathcal{F}\) is defined by an ideal sheaf closed under the Poisson bracket. A maximally coisotropic sheaf corresponds to a Lagrangian subscheme.
   
    Let \( \mathcal{F}\) be a sheaf of \(\mathcal{O}_X \)-modules, defined by a sequence
    \[  0 \rightarrow \mathcal{J} \rightarrow \mathcal{O}_X \rightarrow \mathcal{F} \rightarrow 0. \]
    Recall \( \mathcal{F}\) is coisotropic if \(\mathcal{J}\) is closed or involutive under the Poisson bracket.
    Consider a sheaf \( \mathcal{F}_{\hslash}\). 
    \begin{thm}[\cite{b_defquant}] \label{thm:torsionfree} 
    Suppose the sheaf \( \mathcal{F}_{\hslash}\) of \(\mathcal{A}_{\hslash} \)-modules is defined by an ideal sheaf \( \mathcal{J}_{\hslash} \):
    \[ 0 \rightarrow \mathcal{J}_{\hslash} \rightarrow \mathcal{A}_\hslash \rightarrow \mathcal{F}_{\hslash} \rightarrow 0. \]
    For this quotient to restrict to the sequence for \( \mathcal{F}\), modulo \( \hslash\), such that \( \mathcal{F}_{\hslash}\) is \(\hslash\)-torsion free, then the support of \( \mathcal{F}\) in \(X\) must be coisotropic.
    \end{thm}

    The idea of the proof is to show that given a deformation quantisation, and two elements that restrict to elements of an ideal sheaf, their star product restricts to a Poisson bracket, and then this forces the ideal sheaf to be closed under the Poisson bracket. 

    \begin{proof}
    Consider the subsheaf \(\mathcal{J}_\hslash\)  of \(\mathcal{A}_{\hslash}\) which  maps to \( \mathcal{J}\) in the quotient by \(\hslash \mathcal{J}_{\hslash}\)
    
    Then \( \mathcal{J}_\hslash\) acts on \( \mathcal{F}_\hslash\) via the star product, 
    \( \mathcal{J}_\hslash \otimes_{\mathbf{k} \lBrack \hslash \rBrack } \mathcal{F}_\hslash \rightarrow  \mathcal{F}_\hslash\).  The image of this action is \( \hslash \mathcal{F}_\hslash\).
    
    This is because a section of \( \mathcal{J}_{\hslash}\) can be written as \( f + O(\hslash)\), where \(f \in \mathcal{J}\). Acting on \( \mathcal{F}_{\hslash}\), as \(f\) annihilates any \( O(\hslash^0)\) terms in \( \mathcal{F}_{\hslash}\), so elements of \( \mathcal{F}\), this leaves \( O(\hslash)\) terms which is \(\hslash \mathcal{F}_{\hslash}\).
    
    Define the sheaf \(  [\mathcal{J}_{\hslash},\mathcal{J}_{\hslash} ]:=  \{  f \, \starp_{\pi} \, g - g \, \starp_{\pi}\,  f, f,g \in \mathcal{J}_{\hslash}\}\). Then similarly consider an action 
    \(  [\mathcal{J}_\hslash  ,   \mathcal{J}_\hslash  ]  \otimes_{\mathbf{k}\lBrack\hslash\rBrack} \mathcal{F}_\hslash \rightarrow  \mathcal{F}_\hslash\), which has an image contained in \( \hslash^2 \mathcal{F}_{\hslash}\).
    
    This is because successive applications of \(f\) and \(g\) on \( \mathcal{F}_{\hslash}\), yield \( (f \starp_{\pi} g) \mathcal{F}_{\hslash} = f ( g \mathcal{F}_{\hslash}) \subset f ( \hslash \mathcal{F}_{\hslash} ) \subset \hslash^2 \mathcal{F}_{\hslash}\) (where \( \subset\) means as a subsheaf).

    Finally consider the action by
    \( \frac{1}{\hslash} [ \mathcal{J}_\hslash, \mathcal{J}_\hslash ] \) on \( \mathcal{F}_{\hslash}\). Similary to above, this action has an image \( \hslash \mathcal{F}_{\hslash}\).
    
    Then under the quotient map by \( \hslash \mathcal{F}_{\hslash}\),  \(\hslash \mathcal{F}_\hslash \rightarrow 0\), \( [\mathcal{J}_{\hslash} ,\mathcal{J}_{\hslash}]\) acting on \( \mathcal{F}_{\hslash}\), maps to \( \{\mathcal{J},\mathcal{J}\} \) acting on \( \mathcal{F}\). 
    
    However the image of \(\{\mathcal{J},\mathcal{J}\}\) in the quotient is zero, so this says  \(\{\mathcal{J},\mathcal{J}\}\) annihilates \( \mathcal{F}\).

    Hence \( \mathcal{J} \) is closed under the Poisson bracket.
    \end{proof}
    
    \begin{rem}
    Note \( \mathcal{J}_{\hslash} \neq \mathcal{J}\lBrack \hslash \rBrack\). In some later examples \( \mathcal{J}_{\hslash}(U) =  \sum_i \mathcal{A}_{\hslash}(U) \star  J_i \) where the \(J_i\) are minimal collection of generators from \(\mathcal{J}(U) = \langle J_i\rangle  \), but it is important to note the ring structure is different. \( \mathcal{J}_{\hslash}(U)\) is generated with the star product, while \( \mathcal{J}(U)\) just uses the commutative product.
    \end{rem}

    A further corollary of lemma (\ref{thm:torsionfree}) is that we cannot find an \( \mathcal{A}_\hslash \) \(\hslash \)-torsion free module supported at a point. An example is the sky scraper sheaf associated to a point.

    \begin{ex}Consider the following non-example of theorem (\ref{thm:torsionfree}). Let \( p = (0,0) \in X \cong \mathbb{C}^2 \) be a point. Consider the skyscraper sheaf \( \mathrm{skysc}_p \), defined by  the ideal of functions \( x,y\) with the standard Poisson structure \( \{ x,y \} = 1\). Then in the previous proof we must necessarily have \( \hslash \) torsion.
    \end{ex}

    
    The proof fails if there is \( \hslash\) torsion. If \( \hslash \mathcal{F}_{\hslash} = 0\), then the only possibility is \( \mathcal{F}_{\hslash} = \mathcal{F}\), which is excluded by flatness.
    
    As expected, this torsion condition, lemma (\ref{thm:torsionfree}), is equivalent to the statement that \( \mathcal{A}_{\hslash}\) modules over \( \mathbf{k}\lBrack \hslash \rBrack \) are flat.
    
    \begin{lem} If \( \mathcal{F}_\hslash\) is flat over \( \mathbf{k}\lBrack \hslash \rBrack \), then the support is coisotropic.
    \end{lem}
    
    As a corollary of lemma (\ref{thm:torsionfree}), if the support is isotropic, there must be \(\hslash\)-torsion. In the affine case, of a sheaf on a supported on a subvariety, if \(\mathcal{F}\) is determined by an ideal sheaf, \(\mathcal{J}\), then an example would locally be trying to satisfy an overdetermined collection of equations.

    \section{Cyclic sheaves and cyclic DQ-modules}

    \subsection{Weyl-algebras and D-modules}
    \label{sec:weyl_algebra}
    
    In physics, a operators on a Hilbert space represent the quantisation of some classical theory. For a scheme \( (X,\mathcal{O}_X)\), the algebraic analog of operators on a Hilbert space is a ring of differential operators \cite{k_holonomic}, as a \( \mathcal{O}_X\) module. These rings of differential operators arise from deformation quantisation.

    Let \(X= T^* C\) be a symplectic variety with local Darboux coordinates \( (x_i,y_i)\). Consider the deformation quantisation of \( \mathcal{O}_X\), given by \( \mathcal{A}_{\hslash} = (\mathcal{O}_X\lBrack \hslash \rBrack , \starp_{\pi}) \), and the localisation at \(\hslash\), \( \mathcal{A}_{\hslash}^{\text{loc}} \). \( \mathcal{A}_{\hslash}\) is isomorphic to a sheaf of rings of differential operators on \(X\). These rings are representable as a Weyl-algebra over \( \mathbf{k} \lBrack \hslash \rBrack\). Further the localisation \( \mathcal{A}_{\hslash}^{\text{loc}}\) identifies differential operators with vector fields. 
    

    \begin{defn}[Weyl-algebra over \( \mathbf{k}\)] 
    The Weyl-algebra \( \mathcal{W}_{n}\), is an associative unital \( \mathbf{k}\)-algebra generated by \(n\) elements \( x^{i}\) and \(n\) elements \( \partial_i\) with the constraints that \( x^i x^j = x^j x^i, \partial_i \partial_j = \partial_j \partial_i\) and \( \partial_j x^i - x^i \partial_j  = \delta^{i}_j\).
    \end{defn}
    

    \begin{defn}[\(D\)-module] A \emph{\(D\)-module} is a module over a ring \(D\) of differential operators.
    \end{defn}
    \begin{ex}
    Examples of \(D\)-modules are modules over \( \mathcal{D}(X)\), where \( \mathcal{D}(X)\) is the ring of differential operators on a variety \(X\).
    \end{ex}

    If \( \mathbf{k}\) is characteristic zero, there is an isomorphism between Weyl-algebras, and the ring of differential operators, \(\mathcal{D}( \mathbb{A}_n)\) acting on the ring (or module) \( \mathcal{O}(\mathbb{A}_n)\) of functions of an \(n\) dimensional affine space. Let \( \mathbb{A}_n = \Spec(R), R=\mathbf{k}[x_1, \dots x_n]\). Then
    \begin{lem}
    \(\mathcal{D}(\mathbb{A}_n) \cong \mathcal{W}_n \cong \mathbf{k}[x_1,\dots x_n,\partial_1, \dots , \partial_n].\)
    \end{lem}

    \begin{proof}
    The ring of differential operators \( \mathcal{D}(\mathbb{A}_n) \), is generated by compositions of derivations, \( \mathrm{Der}(R,R)\), on \(\mathbb{A}_n\), which are \(\mathbf{k}\)-linear maps satisfying Leibniz. Derivations on the polynomial ring are given by \( \frac{\partial}{\partial x_i}\), which we identify with \( \partial_i \) in the Weyl-algebra.
    \end{proof}
    
    For a subvariety \(X \rightarrow \mathbb{A}_n\), \(X = \Spec(R)\),  \(R=\mathbf{k}[x_1, \dots x_n]/I\), the ring of differential operators on \(X\) is given as follows:
    \begin{lem}
    \(  \mathcal{D}(X) \simeq \{ d \in \mathcal{W}_n, d(I) \subseteq I \}/I \cdot \mathcal{W}_n,\)  where
    \( \mathcal{W}_n=\mathbf{k}[x_1,\dots x_n, \partial_1, \dots \partial_n]\).
    \end{lem}
    
    \subsection{Weyl-algebras from deformation quantisation}
    
    These Weyl-algebras do not yet give a representation of a deformation quantisation. Consider the case of a \(2n\) dimensional Poisson algebraic variety \(W\), with coordinates \( x^i\) and \(y_i\), so 
    \( \mathcal{O}(W) = \mathbf{k}[x^i, y_i]\). The Poisson bracket is given by \( \{x^i, y_j \} = \delta_{j}^i, \{x_i,x_j\} = \{y_i,y_j \} = 0\). 
    In this case the ring of differential operators \( \mathcal{D}(W) \), is isomorphic to a Weyl-algebra in \(4n\) variables, \( \mathcal{D}(W) \cong \mathcal{D}_{2n}\), where \[\mathcal{D}_{2n} = \langle  x_i, y_i, \frac{\partial}{ \partial x_i}, \frac{\partial}{ \partial y_j} \rangle  .\] 
    Ignoring the issue of \( \hslash \) coefficients, \( \mathcal{D}(W)\) does not give a representation of a deformation quantisation of \( \mathcal{O}(W) \). There are an extra \(2n\) generators. In the \(2n\) dimensional case, a choice of polarisation will give the correct result. Choosing a Lagrangian \(L\) in \(\mathcal{T}_W\), such that \( L = \mathcal{T}_{\cL} \), where \( \cL\) is a subvariety
    determined by an ideal with \(n\) generators, there is an isomorphism \(\mathcal{D}(\mathbb{L})  \cong \mathcal{D}_n\). For a coisotropic subvariety of codimension \(g\),  there is a natural identification with \( \mathcal{D}_{2n-g}\) 
            
    \subsection{Cyclic sheaves}
     Let \(M\) be a module over a commutative ring \(R\). 
    \begin{defn}[Support]
    The support, \( \mathrm{Supp}(M)\) is the set of prime ideals \( \mathfrak{p}\) in \(R\) such that \(M_{\mathfrak{p}} \neq 0\).
    \end{defn}
    
    The following result from commutative algebra is how we will think of cyclic modules. Let \(I\) be an ideal in \(R\).
    \begin{lem}
    \(R/I\) is a cyclic \(R\)-module.
    \end{lem}
    \begin{proof}
    First of all we show \(R/I\) is an \(R\)-module. Scalar multiplication as an \(R\)-module is given by
    \[ r \cdot ( m + I) = r m + I. \]
    Now we show any cyclic \(R\)-module is expressible as \(R/J\). Suppose \(M\) is cyclic. Then there exists an \(m\) such that \( M  = \{ r m , r \in R \}\). Consider the map \(R \rightarrow M\), \( \phi(r) = r m\). The kernel of this map is an ideal \(J\) in \(R\), so \( \phi(r + J) = \phi(r) = r m\). This defines a module homomorphism and a bijection \(R/J \cong M\).
    \end{proof}
    
    \begin{corollary}
    An \(R\)-module \(M\) is cyclic if there exists an ideal \(I\) such that \(M \simeq R/I\).
    \end{corollary}
    This says \(1 + I\) generates \(M\). Converserly, given \( m + I \in R/I\), necessarily \( m + I = m \cdot 1 + m \cdot I = m ( 1 +  I)\). So \(1 + I\) is a generator, \(M = \langle 1 + I \rangle \). Hence \(M\) is cyclic.
    
    In the non-commutative case we can generalise this result for left and right \(R\)-modules.
    
    Now recall some properties of sheaves. Let \( \mathcal{F}\) be a sheaf over a space \(X\). 
    
    \begin{defn}[Stalk of a sheaf] The \emph{stalk} at a point \(x\in X\) is the limit
    \[ \mathcal{F}_x = \colim_{U \rightarrow X | x\in U} F(U).\]
    \end{defn}
    
    The \emph{support of a sheaf} \( \mathcal{F}\) over a topological \(X\),  is similar to the module case, but prime ideals are replaced by points in \(X\) with non zero stalks:
    
    \begin{defn}[Support of a sheaf] 
    The support of a sheaf is the collection of points in \(X\) such that the stalks are non zero, \( \{x \in X, \mathcal{F}_x \neq 0_C \} \).
    \end{defn}
    
    Let \( (X, \mathcal{O}_X)\) be a scheme, and let \(\mathcal{F}\) be a sheaf of \( \mathcal{O}_X\)-modules. Recall a sheaf of \( \mathcal{O}_X\)-modules \(\mathcal{F}\) is defined by functorial compatibility with \( \mathcal{O}_X\) maps:
    \begin{center}
        \begin{tikzcd}
            \mathcal{F}(U) \arrow[r] \arrow[d] & \mathcal{F}(V) \arrow[d] \\
            \mathcal{O}(U) \arrow[r] & \mathcal{O}(V)
        \end{tikzcd}
    \end{center}

    For a sheaf of \( \mathcal{O}_X\) modules to be cyclic, there has to be a single element which generates all modules for all open sets, namely a global section \(H^0(X,\mathcal{F})\).
    
    \begin{defn}[Cyclic sheaf] 
    \label{defn:cyclicsheaf}
    A \emph{cyclic sheaf} of modules \(\mathcal{F}\), is a sheaf where every module
    \(\mathcal{F}(U)\) is cyclic as an \( \mathcal{O}_X\)-module. Further every module is generated by a single element of \(H^0(X,\mathcal{F})\).
    \end{defn}
    
    
    
    \begin{rem} A cyclic sheaf is different to a \emph{sheaf of cyclic modules}. For example
    \( \mathcal{O}_{\mathbb{P}^1}(n)\), is a sheaf of cyclic modules, but is not a cyclic sheaf, since it is not generated by a single section. 
    \end{rem} 
        
    \begin{ex}
    \( \mathcal{O}_X\) is a cyclic sheaf over itself.
    Proof: the constant unit function.
    \end{ex}

    An important example of cyclic arises from a subvariety or a morphism of schemes. Let \(f : Y \rightarrow X\). Consider the pushforward sheaf \( f_{*} \mathcal{O}_Y\). As a \( \mathcal{O}_X\)-module, \( f_{*} \mathcal{O}_Y\) is cyclic if there exists an extension of elements as follows:

    \begin{prop}
    \(f_{*} \mathcal{O}_Y\) is cyclic if given any \(g \in f_{*} \mathcal{O}_Y(U) \), 
    there exists \(h \in \mathcal{O}_X(U)\) such that \(h \cdot 1 = g\).
    \end{prop}
    So for example, in the affine algebraic case, when a regular function \(g\) on \(U \cap f(Y)\) extends to a regular function on \(U\), then \(f_{*} \mathcal{O}_Y\) is cyclic.
    
    \begin{note} Note that often the \(f_{*}\) is omitted if \(Y\) is an actual subvariety, \(Y \rightarrow X\), so if \(f\) is simply the inclusion map, and similarly \( \mathcal{O}_X \rightarrow \mathcal{O}_Y\) is restriction.
    \end{note}
    
    The proposition is true in the affine algebraic case. Consider a subvariety \( Y \rightarrow X\), with  \(\mathcal{O}(Y) = \mathcal{O}(X)/I\) as a \( \mathcal{O}(X)\) module. Pick \(h\) such that \(h \rightarrow g\) in the quotient map \( \mathcal{O}(X) \rightarrow \mathcal{O}(X)/I\).
    

    In the Airy structure setting, this cyclic module encodes a \emph{wavefunction}, which is an object that is used to recover topological recursion. There is a cyclic sheaf, corresponding to a subvariety, that is quantised giving a cyclic \( \mathcal{D}\)-module. The generator of this module is termed a wavefunction.
    
    In general there are going to be extra requirements when a sheaf of modules is quantised (replaced with some non commutative objects).

    \subsection{Cyclic DQ-modules}

    Let \( (X,\mathcal{O}_X)\) be a Poisson scheme (possibly formal), and consider a flat deformation quantisation of \( \mathcal{O}_X\) defined by a sheaf \( \mathcal{A}_\hslash\). What is termed a \emph{cyclic \( DQ\)-module} in \cite{ks_airy}, can be understood from a sheaf theoretic perspective. In particular a non-commutative version of a \hyperref[defn:cyclicsheaf]{cyclic sheaf}. Instead of a cyclic sheaf of \( \mathcal{O}_X\)-modules, there is a cyclic sheaf of \( \mathcal{A}_{\hslash}\)-modules.

    In general finding deformation quantisation modules is a hard task. We construct a particular example relevant for the application to Airy structures.
    
    While quantisation is not an exact functor from the category of sheaves on \(X\) to some category of non-commutative sheaves on \(X\), there are cases where there exists cyclic \( \mathcal{A}_\hslash\)-modules, or a cyclic \(DQ\)-module, that corresponds to cyclic sheaves on \(X\). 
    
    These cyclic DQ-modules naturally arise from the quantisation of subvariety or scheme. Consider a sub-scheme \((Y,\mathcal{O}_Y)\), in \((X, \mathcal{O}_X)\) defined by the sequence:
    \[ 0\rightarrow \mathcal{J} \rightarrow \mathcal{O}_X \rightarrow \iota_*\mathcal{O}_Y \rightarrow 0.\]
    Note \( \mathcal{O}_Y\) is a cyclic sheaf, \( \iota \) is the natural inclusion map. Let \( \mathcal{A}_{\hslash} = (\mathcal{O}_X \lBrack \hslash \rBrack, \star) \) be a deformation quantisation of \( \mathcal{O}_X\), flat over \( \mathbf{k}\lBrack \hslash \rBrack\), with a star product (\ref{defn:star_prod_pois}). Let the ideal sheaf \( \mathcal{J}\) locally be given by a minimal generating set \( \{ H_i\} \), so \( \mathcal{J}(U) = \langle H_i \rangle \). In general, a generator of an ideal \(I\) in a ring \(R\), is a subset \(S\) which generates \(I\).
    
    Furthermore, let \(Y\) be coisiotropic, so \( \mathcal{J}\) is closed under the Poisson bracket. Locally this gives each \( \mathcal{J}(U)\) the structure of a Poisson ideal \cite{jordan}. Define the ideal sheaf \( \mathcal{J}_{\hslash}\) of \( \mathcal{A}_{\hslash}\)-modules:
    \[ \mathcal{J}_{\hslash}(U) = \sum_i \mathcal{A}_{\hslash}(U) \star  (H_i + O(\hslash)), \]
    where \( O(\hslash) \) is some element in \( \hslash \, \mathbf{k} \lBrack \hslash \rBrack\).
    Naturally \( [  \mathcal{J}_{\hslash} ,  \mathcal{J}_{\hslash}] \subset  \hslash \mathcal{J}_{\hslash}\).
    Note that the ideal \( \mathcal{J}_{\hslash}(U)\) uses star multiplications of the \(H_i\). \(J_{\hslash}= \mathcal{J}_{\hslash}(U)\) is an (right) ideal in the non-commutative ring \( A_{\hslash} = \mathcal{A}_{\hslash}(U)\). This means for \( j \in J_{\hslash} \) and \( a \in A_{\hslash}\) we require \( a \star j \in J_{\hslash}\), not some classical product \(a  j\). 
    
    Then consider the sheaf \( \mathcal{E}_{\hslash}\) of \( \mathcal{A}_{\hslash}\) modules  given by 
    \[ 0 \rightarrow \mathcal{J}_\hslash  \rightarrow \mathcal{A}_\hslash  \rightarrow \mathcal{E}_\hslash \rightarrow 0.\]
    \begin{prop} 
    If \(\mathcal{E}_\hslash\) is flat over \( \mathbf{k} \lBrack \hslash \rBrack\), then \( \mathcal{E}_\hslash\) is a deformation quantisation of \(\mathcal{O}_Y\).
    \end{prop}
    
    We construct several examples of this in the following chapters.  Kontsevich and Soibelman \cite{ks_airy} show this works for \emph{quadratic Lagrangians}, where the ideal is given by a collection of quadratic polynomials satisfying certain constraints, and showing the quotient is free over \( \mathbf{k} \lBrack \hslash \rBrack\).
    

    \begin{corollary}
    The sheaf \( \mathcal{E}_{\hslash}\) of \( \mathcal{A}_\hslash\)-modules is supported on \( Y\) modulo \(\hslash\). 
    \end{corollary} 
    We prove this in the affine case, so \(A_{\hslash}=  \mathcal{A}_\hslash(U)\), \( J_{\hslash} = \sum_i A_{\hslash} \star H_i \), \(A =  \mathcal{O}_X(U) \). 
    \begin{proof}
    The ideal \( J_{\hslash}\) restricts to an ideal \( J = \sum A \cdot H_i \cong \langle J_i \rangle \), so we obtain the sequence \( 0 \rightarrow J \rightarrow A \rightarrow A/J \rightarrow 0 \).
    \end{proof}
    
    \begin{corollary}
    The sheaf \( \mathcal{E}_{\hslash}\) of \( \mathcal{A}_\hslash\)-modules is cyclic.
    \end{corollary}
    This simply follows by the definition of \( \mathcal{E}_{\hslash}\).
    
    \begin{ex}[Non example]
    Consider \(J = \langle x,y \rangle \) in \( A=\mathbf{k}[x,y]\), with Poisson bracket \( \{ y,x\}=1\), \( \{ x,x\}=\{y,y\}=0\). Let the deformation quantisation of \(A\) be given by \(A_{\hslash} = (\mathbf{k}[x,y]\lBrack\hslash\rBrack,\star)\), and consider the sum of ideals  \( J_{\hslash}  = A_{\hslash} \star  x + A_{\hslash} \star y \), where \( \star \) is the Moyal product. Now we have
    \[ 0 \rightarrow J_{\hslash} \rightarrow A_{\hslash} \rightarrow \mathbf{k} \rightarrow 0, \]
    but \( \mathbf{k}\) is not flat over \( \mathbf{k} \lBrack \hslash \rBrack\), so this is not a deformation quantisation.
    Also note \( \{y,x\} = 1\), and \(1 \notin J\).
    \end{ex}
    
    

    \section{Wavefunctions}
    \label{sec:wavefunctions} 
    As seen in the preceding sections, deformation quantisation produces a non-commutative algebra of operators. It is not necessarily yet a quantisation as used in physics. The question to ask is, what do these operators act on. In physics this would be analogous to a Hilbert space of functions. In this more algebraic context, it is a dual \( \mathcal{A}_{\hslash}\)-module. These dual modules define the objects known as \emph{wavefunctions}, where the word has the same meaning in quantum curves and topological recursion.  
    
    First, we recall some results from commutative algebra again as inspiration.

    \subsection{A review of annihilators of modules}
    
    Recall some results from commutative algebra:

    \begin{lem} Let \(R\) be a commutative ring, \(I\) an ideal in \(R\). Then for the \(R\)-module \(R/I\), \( \mathrm{Hom}(R/I,R) \cong  \mathrm{Ann}_R(I)\).
    \end{lem} 
    
    \begin{proof}
       Consider the short exact sequence
       \[ 0 \rightarrow I \rightarrow R \rightarrow R/I \rightarrow 0. \]
       Apply the left exact Hom functor \(\Hom_R(\sbt,R)\) gives the left exact sequence:
       \[ 0 \rightarrow \Hom(R/I,R) \rightarrow \mathrm{Hom}(R,R) \rightarrow \Hom(I,R). \]
       So \( \Hom(R/I,R)\) is the kernel of the map \( \Hom(R,R) \rightarrow \Hom(I,R)\). Now \( \Hom(R,R) \cong R\) by the isomorphism given by evaluation at \(1\), \( f \in \Hom(R,R) \rightarrow f(1) \in R\). Similarly for \(\Hom(I,R)\). So explicitly as \(\Hom(R/I,R)\) is the kernel of multiplication of elements of \(I\) by elements of \(R\), \( \Hom(R/I,R) = \{ r \in R | \forall i \in I, \,r \cdot i = 0 \} = \mathrm{Ann}_R(I)\).
    \end{proof}
    
    The same proof follows more generally for any \(R\)-module \(N\) by applying the functor \( \Hom(\sbt, N)\):
    \begin{corollary}
       \(\Hom(R/I,N) \cong \mathrm{Ann}_{N}(I)\).
    \end{corollary}
       
    \begin{lem} There is an isomorphism between the cyclic \(R\)-module \(M\) generated by \(x\), \(M = \langle x \rangle \), and the quotient 
    \( R / \mathrm{Ann}_R(x)\), where \( \mathrm{Ann}_R(x)\) is the ideal formed by all elements in \(R\) that annihilate \(x\):
    \[M =R/\mathrm{Ann}_R(x).\]
    \end{lem}

    A version of these results generalise to the non-commutative case of rings of differential operators. However the handedness of the modules and annihilators will be important.
    
    \begin{defn}[Left and right annihilators]
    Let \( {\lAnn}_R(x) \) denote \emph{left annihilators} of \(x\):
    \[ {\lAnn}_R(x) = \{ l \in R\, | \, l\, x = 0\}. \]
    Similarly let \({\rAnn}_R(x)\) denote \emph{right annihilators} of \(x\):
    \[ {\rAnn}_R(x) = \{r \in R \,| \,x \,r = 0 \}. \]
    \end{defn}     
    

    Let \(R\) be a non-commutative ring or algebra, for example the ring of differential operators.  A \emph{wavefunction} \( \psi\) is a generator of a cyclic right \(R\)-module, \(M  \cong \langle \psi \rangle \cong R/{\lAnn}_R(\psi) \). Also \({\lAnn}_R(\psi) \cong D\), so \( D \psi = 0\).
    
    \begin{ex}\label{ex:sols}
    Let \( \mathcal{O}(X)\) be a ring of smooth functions, \( \mathcal{D}(X)\) the ring of differential operators on \(X\). Consider a differential operator \( D : \mathcal{D}(X)\). Then 
    \[ \mathrm{Hom}_{\mathcal{O}(X)}\left(\frac{\mathcal{D}(X)}{  \mathcal{D}(X)  D } , \mathcal{O}(X) \right)\]  
    represents solutions to \(D\) in \(X\).
    
    \(\mathcal{D}(X)  D \) is a right submodule of differential operators with a factor of \(D\). The module quotient can be represented by \(M=\mathcal{D}(X)/D\).  
    
    Then \( \mathrm{Hom}_{\mathcal{O}(X)}\left(M, \mathcal{O}(X) \right) \cong {\rAnn}_{\mathcal{O}(X)}(D) \cong \langle \psi \rangle\), which are the solutions to the equation 
    \( D \psi = 0\). 
    \end{ex}
    \subsection{Wavefunctions from deformation quantisation}

    Consider a Poisson scheme \((X,\mathcal{O}_X)\), and a deformation quantisation of \( \mathcal{O}_X\) given by \( \mathcal{A}_{\hslash} = (\mathcal{O}_X\lBrack \hslash \rBrack,\star )\), where \( \star\) is a star product. In particular, for the examples where we construct deformation quantisation modules \( \star\) will be a Moyal product per example (\ref{ex:moyal_prod}). 
    
    Consider a coisotropic subscheme, \(Y\) in \(X\), with the map \( \iota : Y \rightarrow X\) giving a sequence:
    \[ 0 \rightarrow \mathcal{J} \rightarrow \mathcal{O}_X \rightarrow \iota_{*}\mathcal{O}_Y \rightarrow 0. \]
    Locally the ideal sheaf is given by a minimal collection of generaors \(\{H_i\}\), \( \mathcal{J}(U) = \langle H_i \rangle \). In finite dimensions if \(X\) is \(2n\) dimensional there will be \(n\) generators or fewer. In the infinite dimensional case, where \(Y\) is a quadratic Lagrangian, for Airy structures there will be a special choice of \(H_i\).
    
    Consider a sheaf \( \mathcal{E}_{\hslash}\) of \( \mathcal{A}_{\hslash}\)-modules corresponding to the deformation quantisation of the coisotropic subscheme \( \iota_{*} \mathcal{O}_Y\), so
    \[ 0 \rightarrow \hslash \mathcal{E}_{\hslash} \rightarrow \mathcal{E}_{\hslash} \rightarrow \iota_{*} \mathcal{O}_Y \rightarrow 0.\]
    Further suppose \( \mathcal{E}_{\hslash}\) is defined by a quotient:
    \begin{equation}
        \label{eqn:exactseqJAE}
        0 \rightarrow \mathcal{J}_{\hslash}   \rightarrow \mathcal{A}_\hslash  \rightarrow \mathcal{E}_\hslash \rightarrow 0,
    \end{equation} 
    where locally \( \mathcal{J}_{\hslash}(U) = \sum_i \mathcal{A}_{\hslash }(U) \star  H_i  \) is given by the generators from \( \mathcal{J}\).
    Similarly, \( \mathcal{J}_{\hslash}/\hslash \mathcal{J}_{\hslash} = \mathcal{J} \). Consider the dual \( \mathcal{A}_{\hslash}\) module:
    \[ \mathcal{E}_{\hslash}^{\vee} = \mathrm{Hom}_{\mathcal{A}_{\hslash}}(\mathcal{E}_{\hslash},\mathcal{O}_X \lBrack \hslash \rBrack ).\]
    Necessarily there may be conditions on the existence of such an object, for example in  \cite[section 2.4, page 15]{abpolyquant}. If this module exists, it represents annihilators of the generators of \( \mathcal{J}_{\hslash}\). For example let \(H \in \mathcal{J}_\hslash\), then  \( w \in \mathcal{E}^{\vee}_{\hslash}\) are sections such that 
    \begin{equation}
        \label{eqn:annih}
        H \star w = 0.
    \end{equation}
    Consider an isomorphism  \( \varphi : \mathcal{A}_{\hslash} \rightarrow \mathcal{W}_{\hslash}\) where \( \mathcal{W}_{\hslash}\) is a Weyl-algebra. \( \varphi\) sends the non-commutative star product of functions to the non-commutative composition of operators.
    \begin{ex} 
    For example let \(A=\mathbf{k}[x_1,\dots,x_n,y_1,\dots y_n]\), be equipped with the Poisson bracket:  \[\{y_i,x_j\} = \delta_{ij}, \quad  \{x_i,x_j \} = \{ y_i,y_j\} = 0,\] and consider the deformation quantisation \( \mathcal{A}_{\hslash} = (\mathbf{k}[x_1,\dots,x_n,y_1,\dots y_n]\lBrack\hslash \rBrack, \star )\).
    Then there is an ismorphism, \( \varphi\), between \( \mathcal{A}_{\hslash}\) and the Weyl-algebra \( \mathcal{W}_{\hslash} = \mathbf{k} [x_1, \dots, x_n , \hslash \partial_1 , \dots \hslash \partial_n ] \lBrack \hslash \rBrack\). The map \( \varphi\) is defined on the generators by:
    \[ \varphi(x_i \star x_j) = x_i \cdot  x_j, \quad \varphi( y_i \star y_j) = \partial_i \cdot \partial_j, \quad \varphi(x_i \star y_j ) = x_i \cdot \hslash \partial_j. \]
    \end{ex}
    
    Under the isomorphism \( \varphi\) to the Weyl algebra, \( \varphi\) sends the star product equation (\ref{eqn:annih}), to differential equation
    \[ \varphi(H \star w ) = \varphi(H) \cdot \varphi(w) = 0,\]
    where \( \varphi(H)\) is an operator. 
    
    Furthermore \( \varphi\) defines isomorphisms \(  \mathcal{J}_{\hslash}\cong \varphi(\mathcal{J}_{\hslash})\) and \( \mathcal{E}_{\hslash} \cong \varphi ( \mathcal{E}_{\hslash})\), which commute at each step of the sequence in equation (\ref{eqn:exactseqJAE}), so 
    \[ 0 \rightarrow \varphi(\mathcal{J}_{\hslash}) \rightarrow \mathcal{W}_{\hslash} \rightarrow \varphi(\mathcal{E}_{\hslash}) \rightarrow 0, \]
    is also exact, via diagram chasing. 
    
    In previous work, for example in \cite{norbury_quant, ks_airy}, these modules are examined at a point, or at a formal completion, from the Weyl-algebra perspective. This is primarily to solve the differential equation with the WKB method at that point.
    Globally \( \mathcal{E}_{\hslash}^{\vee}\) does not have to be cyclic, or generated by a single section. In the example (\ref{ex:the_conic}), \( \mathcal{E}_{\hslash}^*\) is defined as solutions to the Airy equation, which is dimension two. A general question in algebra is how to reconstruct more global information from infinitesimal local information, for example the Artin reconstruction theorem.
    
    Let \(p\) be a point in \(X\), with a corresponding ideal sheaf \( \mathfrak{m}\). In example (\ref{ex:the_conic}), this will be the point \(0\). Denote \( \widehat{\mathcal{A}}_{\hslash}\) as completion along \( \mathfrak{m}\). Similarly we consider a completed Weyl-algebra \( \widehat{\mathcal{W}}_{\hslash}\). 
    
    \begin{defn}[Wavefunction]
    A \emph{wavefunction} \(w\) is an element of the module 
    \[ w \in \mathrm{Hom}_{\mathcal{A}_{\hslash}} (\widehat{\mathcal{E}}_{\hslash},  \widehat{\mathcal{O}}_{X}\lBrack \hslash \rBrack ),\] 
    or alternatively in the Weyl-algebra perspective \( \psi\), 
    \[ \psi \in  \mathrm{Hom}_{\mathcal{W}_{\hslash}} (\varphi(\widehat{\mathcal{E}}_{\hslash}),  \widehat{\mathcal{W}}_{\hslash} ).\]
    \end{defn}

    Of course the global objects could also be called a wavefunction.  This definition is to be consistent with others in the literature, for example the wavefunctions in \cite[page 52]{ks_airy}, which are represented as a formal power series found by the WKB method at a point in \(X\). Alternatively called \( \mathcal{D}\)-modules in \cite[page 148]{holland}.
    The formal power series also often contains interesting enumerative data, for example in the case of the conic, example (\ref{ex:the_conic}), the wavefunction is the generating function for rooted planar graphs in various genus. The important thing we emphasise is that these module in these examples all arise from deformation quantisation, and are not constructed in an ad-hoc fashion.
    
    There is the possibility that there are other deformation quantisation modules which restrict to \(\iota_{*} \mathcal{O}_Y\) under the quotient by \(\hslash\). We are not ruling out any possibilities, only constructing modules in specific cases.
    Kontsevich and Soibelman find a particular generator in the case where \(Y\) is a quadratic Lagrangian \cite{ks_airy}. 
    

    The deformation quantisation perspective highlights that \emph{wavefunctions} as understood in topological recursion and Airy structures are a module homomorphism. 
    We want to emphasise the importance of the deformation quantisation perspective, as these modules may be more interesting than just a particular choice of generator.
    
    The other primary use of the deformation quantisation perspective is to understand the origin of the intersection formula for wavefunctions presented in \cite[section (7.2)]{ks_airy}. The formula presented there gives a wavefunction on the deformation space of curves. The concept of wavefunction reduction can be understood as an intersection of sheaves which we explore in section (\ref{section:wavefunction_reduction}).

%% file: conic.tex
 \subsection{Deformation quantisation of the conic}
    \label{sec:def_of_conic}
    
    We give an explicit construction of a deformation quantisation module, and an associated wavefunction in the case of a conic in the plane.
    
    Let \( W = \mathbb{A}^2\), \(\mathcal{O}(W) = \mathbf{k}[x,y]\), with Poisson bracket 
    \[ \{y,x\}=  (dx\otimes dy) ( \Pi) = \pi_{ij} \partial_i x \partial_j y = 1, \quad \{ x,x\} = \{y,y\}=0 \]
    and for clarity \(i,j \in \{x,y\}\), \( \pi_{xy}=-1,\pi_{yx}= 1,\pi_{xx} = \pi_{yy}=0\), or alternatively as a matrix: 
    \[ \pi = \left( \begin{array}{cc}
         0 &-1  \\
         1 & 0 
    \end{array}\right).\]
       \begin{ex}
       \label{ex:the_conic}
        Consider the conic \( \cL \subset W\), defined by the quotient \(\mathcal{O}(\mathbb{L})= \mathbf{k}[x,y]/\langle H = -y + x^2 + 2 x y + y^2\rangle \).
    \end{ex}
    Wavefunctions associated to \( \mathbb{L}\) are defined on a formal neighbourhood of the origin \((x,y)=(0,0)\).
    
    Corresponding to \(\cL \) is a formal scheme \( \widehat{\cL \hspace{0pt} } \), (respectively \(W\)), which is given by completion of \(\cL \) by a maximal ideal \( \mathfrak{m}=\langle x, y\rangle\). This is given by taking the colimit of the quotient: 
    \[ \widehat{\cL} = \colim_n  \mathrm{Spec} \left( \mathcal{O}(\cL )/\mathfrak{m}^n \right). \]
    \( \widehat{\cL}\) is representable by a formal series \(u_0(x)\), so \( y(x) =u_0(x)\), 
    \[ \widehat{\cL} = \Spf \left(\mathbf{k}\lBrack x,y\rBrack /\langle y - u_{0}(x \rangle \right)\] 
    Where  \( u_0(x) \in \mathbf{k}\lBrack x \rBrack\) is found by taking \( -y + x^2 + 2 xy + y^2 = 0 \) and solving for \(y(x)\). \(u_0(x)\) has a closed form, understood as a formal power series, given by: 
    \[ u_0(x) = \frac{1}{2}\left( 1 - \sqrt{1-4x} - 2 \, x \right),\]
    in a formal neighbourhood of \(y=x=0\), 
    \begin{rem}
    The positive square root would correspond to completing around \(y=1\).
    \end{rem}
    Expanding, \( u_0(x)\) is the generating function for Catalan numbers:
    \[ u_0(x) = x^2 + 2 x^3 + 5 x^4 + \cdots + \frac{(2n)!}{(n+1)!\,n!} x^{n+1}  + \cdots,\]
    which counts rooted binary trees.
    
    \(\mathbb{L}\) is locally representable in a formal neighbourhood of the origin as a primitive \( S_0\), so
    \[ y(x) dx = u_0(x) dx = d S_0(x).\] 
    Consider the deformation quantisation of \( \mathcal{O}(W)\) given by \( \mathcal{A}_{\hslash} = (\mathbf{k}[x,y]\lBrack \hslash\rBrack,\star)\), where \( \star\) is the Moyal product per example (\ref{ex:moyal_prod}). Correspondingly, after completion along \( \langle x , y \rangle\) there is a deformation quantisation of the formal scheme associated to \(\widehat{W}= \mathrm{Spf}( \mathbf{k}\lBrack x,y \rBrack) \), which gives the formal non-commutative algebra 
    \begin{equation}
    \label{eqn:completed_conic}
    \widehat{\mathcal{A}}_{\hslash} = ( \mathbf{k} \lBrack x,y \rBrack \lBrack \hslash \rBrack , \star).
    \end{equation}
    Duals to these formal completed algebras are used to investigate the wavefunctions.
    
    \subsubsection{The deformation quantisation of the conic, example (\ref{ex:the_conic}), on formal neighbourhoods}
    The corresponding deformation quantisation \(\widehat{\mathcal{A}}_{\hslash} = ( \mathbf{k} \lBrack  x,y \rBrack \lBrack \hslash \rBrack , \star) \) of example (\ref{ex:the_conic}), (on the formal completion) is representable by the Weyl-algebra \( \widehat{\mathcal{W}}_{\hslash}=(\mathbf{k} \lBrack x, \hslash \partial_x \rBrack \lBrack \hslash \rBrack , \cdot) \), where \( \cdot\) is ordinary composition, as follows. First, checking some terms of the star product:
    \[ x \star x = x^2 , \quad  x \star y = x y - \frac{1}{2} \hslash , \quad y \star y = y^2, \quad y \star x = x y + \frac{1}{2}\hslash, \]
    where it is important to note a term like \(xy\) is another function in
    \( \mathbf{k}[ x,y ]\lBrack\hslash \rBrack\). Checking the anti-commutator:
    \[ \frac{1}{\hslash}\left( \hat{y} \star \hat{x} - \hat{x} \star \hat{y} \right) = 1 =  \{y,x\} , \]
    which is the Poisson bracket as expected. Now there is a map between \( \mathcal{A}_{\hslash}\) and the Weyl-algebra \(\mathcal{W}_{\hslash}\). The map \( \varphi  :  \mathcal{A}_{\hslash}   \rightarrow \mathcal{W}_{\hslash}\),  is required to have the property:
    \[ \varphi( f \star g ) = \varphi(f) \cdot \varphi(g).\] 
    On monomials this is given by:
    \begin{align*}
        \varphi(x) &= x, \\ 
        \varphi(y) &= \hslash \partial_x. 
    \end{align*}
    This is the source of the quantisation rule \( x \rightarrow x\), \( y \rightarrow \hslash \partial_x\).
    
    The map \(\varphi\) is an isomorphism. It sends star products of functions to products of operators.  For example \( x \cdot ( \hslash \partial_x ) = \varphi( x \star y)\). An operator like \( \varphi(xy)\) is defined to be \( \varphi(x \star y) + \frac{\hslash}{2}\), or \( \varphi(y \star x ) - \frac{\hslash}{2}\).  The star product was defined to account for the potential ambiguity. The action of both ways of writing the operator \( \varphi(xy) \) on functions are equivalent. 
    
    This map is analogous to the Wigner-Weyl transform. The star-product is analogous to constructing the phase-space picture of quantum mechanics, while the Weyl-algebra is more analogous to operators in the Schrodinger picture from quantum mechanics \cite{baker}.
    
    Now we  consider a module supported on the deformation quantisation of  \( \mathcal{O}(\mathbb{L})\) in \( \mathcal{O}(W)\). The purpose is to study \(\mathcal{A}_{\hslash}\)-modules, which under the quotient map by \(\hslash^2\), give the \(\mathcal{O}(W)\)-module \(E =\mathcal{O}(W)/ \mathcal{O}(W) H  \). Consider elements of \( \mathbf{k}[x,y]\lBrack\hslash\rBrack\), of the form \(H - \hslash J\). 
    Then the modules of interest are
    \begin{equation}  
    E_{\hslash} = \frac{\mathcal{A}_{\hslash}}{  \mathcal{A}_{\hslash} \star (H - \hslash J)  } ,
    \end{equation}
    for different \(J \in \mathbf{k}[x,y]\lBrack\hslash\rBrack\).  Note that this is a quotient defined via a new ideal \( \mathcal{A}_{\hslash} \star (H - \hslash J) \), constructed from the generator \(H\). It is not a quotient by the classical ideal \( \langle H \rangle\). This module corresponds to operators acting on functions on \( \mathbb{L}\). Also we consider the dual as a \( \mathcal{A}_\hslash\)-module:
    \begin{equation} 
    E_{\hslash}^{\vee} = \mathrm{Hom}_{ \mathcal{A}_{\hslash} }\left(E_{\hslash}, \mathcal{A}_{\hslash} \right),
    \end{equation}
    which represents solutions to the equation \( (H - \hslash J) \star w =0\), for \( w \in  E_{\hslash}^{\vee}\). Correspondingly we can investigate the completions \( \widehat{E}_{\hslash}\) and its dual, but with corresponding Weyl-algebras. 
    
    The purpose of the formal completion is a calculation tool to be able to write down a formal series representation for \(w\). The formal series for \(w\) can contain interesting enumerative data. 
    
    \begin{rem}
    \label{rem:constraints}
    There are possible constraints on the \(J\). Consider a higher dimensional example, where  \( \mathcal{A}_{\hslash}= \mathbf{k}[x_{\sbt},y_{\sbt}] \lBrack \hslash \rBrack\) with the star-bracket \( [f,g] = f \star g - g \star f \). Let \( \mathcal{J} = \{H_i(x_{\sbt},y_{\sbt})\} \) be an ideal closed under the Poisson bracket. Consider the deformation \(\mathcal{J}_{\hslash} = \{ H_i + \hslash J_i \}\). To satisfy theorem (\ref{thm:torsionfree}), \(\mathcal{J}_{\hslash}\) is also required to be closed, \( [\mathcal{J}_{\hslash}, \mathcal{J}_{\hslash} ] \subseteq \mathcal{J}_{\hslash}\). In the case where \(H_i\) are degree two in \(x_{\sbt}\) and \( y_{\sbt}\), one possibility is  find \( J_i\) as a function of only \(\hslash\), for example the quantisation of an Airy structure in \cite{ks_airy}.
    \end{rem}
    
    \begin{rem}
    Deformation quantisation does not uniquely determine \( \mathcal{A}_{\hslash}\)-modules. For a given \( \mathcal{O}_W\)-module \(M\), there are multiple \( \mathcal{A}_{\hslash}\)-modules \(M_{\hslash}\) such that \(M_{\hslash}/\hslash M_{\hslash} = M\).
    \end{rem}
    This two dimensional case is less restricted than the example in remark \ref{rem:constraints}. However we consider when \(J\) is just a function in \(\hslash\), so elements of the module represent solutions to the equation
    \[ H \star w_{\mathbb{L}}(x) = \hslash j(\hslash) w_{\mathbb{L}}(x), \]
    where \( j(\hslash) \) is the function \( j(\hslash) = j_1 \hslash  + \mathcal{O}(\hslash^2)\). In the classical limit, \(H - \hslash J\) restricts to \( H\), and the equation becomes \( H w =0\), which is describing an element of the module 
    \[ w \in \mathrm{Hom}_{\mathcal{O}(W)}(\mathcal{O}(W)/  \mathcal{O}(W) H , \mathcal{O}(W) ) = \mathrm{Ann}_{\mathcal{O}(W)}(  H) =  0,\]
    which is just zero, as zero is the only annihilator of \(H\).

    Using the isomorphism to the Weyl-algebra, this becomes an equation involving operators:
    \begin{equation} \varphi(H) \cdot \varphi(w_{\mathbb{L}}) = \varphi(H) \cdot \psi_{\mathbb{L}} = \hslash j \psi_{\mathbb{L}},
    \end{equation}
    where 
    \[ \varphi(H)  = - \hslash \partial_x +  x^2 + 2  \hslash\, x \, \partial_x + \hslash^2  \partial^2_x + \hslash. \]
    In the Weyl-algebra representation, the module \(\widehat{E}_{\hslash}\) is a quotient of the Weyl-algebra \( \widehat{\mathcal{W}}_{\hslash} =  (\mathbf{k} \lBrack x, \hslash \, \partial_x \cdot \rBrack \lBrack \hslash \rBrack ,\cdot) \). In particular the elements of \( \widehat{E}_{\hslash}^\vee\), \( \psi_{\mathbb{L}}\) are called   \emph{wavefunctions}, \cite[page 13]{ks_airy}. In this case, \( \widehat{E}_{\hslash}^\vee\) (for a fixed \(J\)) is one dimensional, so \( \psi_{\mathbb{L}}\) is a generator of a cyclic module:
    \[ \langle \psi_{\cL} \rangle = \Hom_{\mathcal{W}_{\hslash}}\left(  \frac{\widehat{\mathcal{W}}_{\hslash}}{\widehat{\mathcal{W}}_{\hslash} \langle - \hslash \partial_x +  x^2 + 2  \hslash\, x \, \partial_x + \hslash^2  \partial^2_x + \hslash - \hslash j  \rangle }, \widehat{\mathcal{W}}_{\hslash} \right).\] 
    This module is supported on formal neighbourhoods, \( \widehat{\cL}-\{\hslash=0\}\), and elements can be viewed as the formal  solutions to the equation: 
    \begin{equation} 
    \label{eqn:quant_conic}
    \left( - \hslash \frac{d}{d x} +  x^2 + 2  \hslash\, x \frac{ d }{d x}+ \hslash^2  \frac{d}{d x}  + \hslash - \hslash j\right) \psi_{\mathbb{L}}(x)  = 0, 
    \end{equation}
    where \(  \frac{d}{dx} = \partial_x\), at \(x=0\). Writing this as an eigenvalue equation:
    \[ \left(- \hslash \frac{d}{d x} +  x^2 + 2  \hslash\, x \frac{ d }{d x}+ \hslash^2  \frac{d}{d x} \right) \psi_{\mathbb{L}}(x)  = \hslash ( j -1) \, \psi_{\mathbb{L}}(x). \]
    Consider when \(j = 0\), \( \psi_{\mathbb{L}}(x)= \psi_{\mathbb{L},0}(x)\). Explicitly \( \psi_{\mathbb{L},0}(x)\) is found via a formal series of the form
    \[\psi_{\cL,0}(x) = \mathrm{const} \, \exp\left(\frac{1}{\hslash} S(x)\right),\]
    where \( S(x) = \sum \hslash^g S_g(x)\), and const is a constant. 
    \begin{rem} Note, in other work such as \cite{ks_airy}, the solution is usually chosen to correspond to the point \(x=y=0\), which is a choice of square root. Normally a second order differential equation such as equation (\ref{eqn:quant_conic}) has a two dimensional space of solutions. Completing at \( \langle x, y\rangle\) is analogous to picking a choice of square root. Completing \( \mathcal{W}_{\hslash}\) only along \(x\), so \( \mathbf{k}[\hslash \partial ] \lBrack x\rBrack \lBrack \hslash \rBrack\) , and then expanding as a power series in \(x\) would be analogous for looking for both solutions. In this case the other solution would correspond to the point \(x=0, y = 1\). 
    \end{rem}
    
    The \( S_g(x)\) satisfy \emph{abstract topological recursion}, \cite{ks_airy}, as the conic is a finite dimensional example of a  \emph{Airy structure}, as defined in \cite{ks_airy}. Writing out some terms, where integration means formal integration of polynomials: 
    \[ S_0(x) = \int dx\, u_0(x), \quad S_1(x) = \int dx\, u_1(x),\]
    where \[u_1(x) =\frac{2 u_0(x)  x}{1 - 4 x} = 2 x + 10 x^2 + \cdots + \left(4^n - \frac{(2n)!}{(n!)^2}\right) x^n + \cdots \] is the generating function counting genus one Feynmann diagrams or rooted maps \cite{feynmaps, walsh}.
    
    \begin{prop}
    \(u_{g}(x) \) is the generating function for numbers of genus \(g\) Feynmann diagrams.
    \end{prop}
    Note these are formal objects defined at the point \(y=x=0\), (which is a point on \( \mathbb{L}\)).
    
    \subsubsection{Solving the star-product equation directly}
    
    
    The purpose of the formal completion of \( \mathcal{A}_{\hslash}\), in \(x,y\), \( \widehat{\mathcal{A}}_{\hslash} = \mathbf{k}\lBrack x, y\rBrack \lBrack \hslash \rBrack\),  is to write down a formal power series solution to the differential equation in a series of \(x\) and \(\hslash\). Further this solution is chosen at a single point.
    
    However we can investigate the star-product equation \( H \star w =0 \), and hope to find a function for \(w\) as a series in \(\hslash\), where 
    \[ H = -y + x^2 + 2 x y + y^2.\]
    As \(H\) is degree \(2\), there are only two terms in the star-product to consider:
    \begin{align*}
        H \star w &= H w + \frac{\hslash}{2} \mathrm{prod} \cdot (\pi_{ij} \partial_i ( H ) \otimes \partial_j (w) ) + \frac{\hslash^2}{4} \mathrm{prod} \cdot ( \pi_{ij} \pi_{kl} (\partial_i \partial_k H) \otimes (\partial_{j} \partial_{l} w )), \\
        &= H w +\frac{\hslash }{2}  \left((2 x+2 y-1)  \frac{\partial w}{\partial x}-(2 x+2 y) \frac{\partial w}{\partial y} \right) +  \frac{\hslash^2 }{4} \left(2 \frac{\partial^2 w}{\partial^2 y}-4 \frac{\partial^2 w}{\partial x \partial y}+2 \frac{\partial^2 w}{\partial x^2}\right)=0.
    \end{align*}
    Instead of a second order ordinary differential equation in the completed case, equation (\ref{eqn:quant_conic}), in terms of \(x\), this is a second order PDE in terms of \(x\) and \(y\). 

    Suppose we have a solution of the form \(w(x,y)= \lambda(H(x,y))\). The entire equation reduces to
    \[ \lambda(H) H + \frac{\hslash^2}{2} \frac{\partial^2 \lambda(H)}{\partial^2 H } = 0,\]
    which is similar to the \emph{Airy equation} \cite{airy}, or the equation for a harmonic oscillator.  Writing out a series solution for \( \lambda\) in powers of \( 1/\hslash\) and \(H\), around \(H=0\), gives
    \[ \lambda(H) = \sum_m \frac{\lambda_{3m+4} }{\hslash^{3m+4}} H^{3m+4},\]
    where
    \[ \lambda_{n+2} = \frac{-1}{(n+1)(n+2)} \lambda_{n-1},\]
    and \( \lambda_0\), and \( \lambda_1 \in \mathbf{k}\). In this case the solution \(w(x,y)\) is an element of the module \( \mathbf{k}[x,y]\lBrack\hslash^{-1} \rBrack\). 
    \begin{ques} 
    Can the \(S_g(x)\) be recovered from \( \lambda\)?
    \end{ques}
    

%% file: reduction.tex
 \section{Symplectic reduction}

    First we recall the definition  of symplectic reduction from differential geometry. Let \( (M,\omega) \) be a symplectic manifold. Let \(G\) be a Lie group acting on \(M\), with a corresponding Lie algebra \( \mathfrak{g}\). Consider the moment map \( \mu : M \rightarrow \mathfrak{g}^*\).  For any \( \xi : \mathfrak{g}\), and a point \( p : M\), define a vector field on \(M\):
    \[ X_\xi(p) := \left. \frac{d}{dt} \exp( t \xi)  \cdot p \, \right|_{t\rightarrow 0}.\] 
    The moment map satisfies the property that:
    \[ \omega \cdot X_\xi = d \langle \mu, \xi \rangle. \]
    
    \begin{defn}[Symplectic reduction]
    \label{defn:symplredmani}
    The \emph{symplectic reduction} of \(M\) by \(G\), denoted by \(M \sslash G \), is defined as the quotient 
    \[M \sslash G := \mu^{-1}(0) / G .\] 
    This quotient is a topological quotient given by identifying points in the same orbit.
    \end{defn}
    
    \begin{lem}
    The symplectic form on \(M\) defines a symplectic form on the quotient \(M \sslash G\).
    \end{lem} 
    
    \subsection{Symplectic reduction of vector spaces}
    
    The symplectic reduction of manifolds motivates reduction of symplectic vector spaces, Tate spaces and more generally affine spaces by an \emph{isotropic} subspace.
    
    Let \( (W, \omega)\) be a strong symplectic vector (or Tate) space. Let \( G\) be a linear coisotropic subspace, and consider the set \( G^\perp := \{ w  \in W | \omega(w,g)=0, \; \forall g \in G\} \). As \(G\) is coisotropic \(G^{\perp} \subset G\).
    
    The idea is to treat \(G^\perp\) as a group acting on \(W\). To do this consider a representation of \(G^\perp\) in \(\mathrm{Aff}(G)\), \(G^\perp \rightarrow \mathrm{Aff}(G)\). This subgroup acts on \(W\) by translation, \( W \times G^{\perp} \rightarrow G : (w,g^{\perp} ) \rightarrow w + g^{\perp} \). 
    
    This action preserves the symplectic form \( \omega\). Then we identify a moment map \( \mu\) via the exact sequence
    \[ 0 \rightarrow G \rightarrow W \stackrel{\mu}{\rightarrow} (G^\perp)^* \rightarrow 0,\]
    so \(G = \mu^{-1}(0)\). 
    
    \begin{defn}[Symplectic reduction of vector spaces] The \emph{symplectic reduction} of \(W\) by \(G^{\perp}\), \(W \sslash G^{\perp}\), is defined as the quotient 
    \[ W \sslash G^{\perp} := G/G^\perp = \mathcal{H}. \]
    \end{defn}
    
    \begin{lem}
        The symplectic form on \(W \sslash G^{\perp} \) is well defined.
    \end{lem}

    Let  \( \mathbb{L} \subset W\) be a Lagrangian sub vector space.
    
    \begin{lem}[Reduction of Lagrangians] 
    The Lagrangian \( \mathbb{L}\) is mapped to a Lagrangian \( \mathcal{B}  \subset  W \sslash G^{\perp} \). 
    \end{lem}
    
    So we are being asked to consider the image of \( \mathbb{L}\) in the map \( W \rightarrow G/G^{\perp}\).
    

    \subsection{Symplectic reduction of schemes}
    We now define a sheaf theoretic version of of symplectic reduction. We call this \emph{Poisson reduction}, because it is on the level of algebras. 
    
    Let \((W,\mathcal{O}_W)\) be a Poisson scheme over \( \mathrm{Spec}(\mathbf{k})\), where \( \mathbf{k}\) is a field of characteristic zero. Let \( \cL\) be a Lagrangian in \(W\). Consider \(G\), a codimension \(g\), linear, coisotropic subscheme in \(W\).
    The quotient \(W \sslash G^{\perp} \) is described sheaf theoretically as follows:
    \begin{defn}[Poisson reduction]
    The \emph{Poisson reduction} of \( \mathcal{O}_W\) by \(G^{\perp}\), is the sheaf \( \mathcal{O}_G^{G^{\perp}}\) of \( \mathcal{O}_W\)-modules. \( \mathcal{O}_G^{G^{\perp}}\) is the sheaf of invariants defined by \(G^{\perp}\) action on \( \mathcal{O}_G\).
    \end{defn}
    
    The next question is how to obtain the image of \( \mathbb{L}\) in \( \mathcal{H} = G/G^{\perp}\). Define \( \mathcal{O}_X = \mathcal{O}_W \otimes_{\mathbf{k}} \mathcal{O}_G^{G^\perp}\).
    The image of \( \mathbb{L}\) is obtained with the following fibre product:
    
    \begin{defn}
    The reduction of \( \mathcal{O}_{\cL} \), by \( G^{\perp}\), (when it exists) is the sheaf \( \mathcal{O}_\mathcal{B}\) of \( \mathcal{O}_X\)-modules that makes the following diagram commute:
    \begin{center}
        \begin{tikzcd}
        \mathcal{O}_X \arrow[r] \arrow[d] & \mathcal{O}_G \arrow[d]\\
        \mathcal{O}_{\cL} \arrow[r] &  \mathcal{O}_{G \cap {\cL}}   \cong \mathcal{O}_{\mathcal{B}}
        \end{tikzcd}
    \end{center}
    \end{defn}
    Note that this gives a copy of \( \mathcal{O}_{\mathcal{B}}\) inside \( \mathcal{O}_X\).

    \begin{prop} \( \mathcal{O}_{\mathcal{B}}\) exists when \( \mathbb{L}\) and \(G\) intersect transversally. 
    \end{prop}
    We give an example in section (\ref{sec:examplereduct4d}).  Note that intersections in general may not be defined, for example in section (\ref{sec:examplereduct4d}) transversality is explicitly required.
    
    When \(G\) and \( \mathbb{L}\) are defined by ideal sheaves \( \mathcal{J}_G\) and \(\mathcal{J}_{\cL}\) then necessarily \( \mathcal{O}_B\) is defined as the sum of the ideals \( \mathcal{J}_G + \mathcal{J}_{\cL}\). However this is as a \( \mathcal{O}_X\)-module.
    
    As before, define \( \mathcal{H} = G/G^{\perp}\). We want to relate \( \mathcal{O}_\mathcal{B}\) as a \( \mathcal{O}_X\) to a \( \mathcal{O}_{\mathcal{H}}\)-module. To do this we look for a map \( \mathcal{O}_X \rightarrow \mathcal{O}_{\mathcal{H}} \). In finite dimensions, for example (\ref{sec:examplereduct4d}), this it is even possible for this map to be algebraic. However in the case examined in \cite{chaimanowong2020airy}, \( W\) is replaced by an infinite dimensional space. This map no longer is algebraic, and must necessarily be formal.



    \subsection{Symplectic reduction of the wavefunction}
    \label{section:wavefunction_reduction}
    
    Now we examine how symplectic reduction interacts with deformation quantisation.  Let \( (W,\mathcal{O}_W)\) be an affine symplectic space, \(G\) and \( \mathbb{L}\) a coisotropic and Lagrangian subscheme in \(W\) respectively. Suppose the Poisson reduction of \( \mathcal{O}_W\) by \(G^{\perp}\) exists, and denote it \( \mathcal{O}_{\mathcal{H}}\). Likewise denote the reduction of \( \mathcal{O}_{\mathbb{L}}\) by \( \mathcal{O}_{\mathcal{B}}\). Consider the extension of \(W\), denoted \(X\), so the coisotropic space \(G\) is Lagrangian in \(X\), denoted \(G_X\).
    
    
    There appears to be two interesting wavefunctions defined on \(  \mathcal{B}\). The first is a wavefunction \( \psi_{\overset{\sim}{\mathcal{B}}}\), from the deformation quantisation of  \(\mathcal{O}_{\cB}\) inside \( \mathcal{O}_\mathcal{H}\). The second is defined by Kontsevich and Soibelman in \cite[p. 53]{ks_airy}.
    
    For the first definition, this means taking \( \mathcal{O}_X\), and looking for a deformation quantisation module supported on the formal neighbourhoods of the intersection \(  G \cap \mathbb{L} \). 
    
    From the fibre product, suppose the intersection \( \mathcal{O}_{G \cap \mathbb{L} }\) is defined by an ideal sheaf \( \mathcal{J}_{G \cap \mathbb{L} } = \mathcal{J}_G + \mathcal{J}_\mathbb{L}\),
    \begin{equation}
        \label{eqn:seqforB}
        0 \rightarrow \mathcal{J}_{G \cap \mathbb{L}} \rightarrow \mathcal{O}_X \rightarrow \mathcal{O}_{G \cap \mathbb{L}} \rightarrow 0.
    \end{equation}
    Then take the equations for \(G\) and \(L\), represented by \( \mathcal{J}_{G \cap \mathbb{L}} = \langle P_i \rangle\), and pick a Weyl-algebra representation of the deformation quantisation of these equations, which gives a collection of differential operators, \( \langle \hat{P}_i - J \hslash \rangle \). \(\psi_{\overset{\sim}{\mathcal{B}}}\) is the WKB solution to the resulting differential equations:
    \[ (\hat{P}_i - \hslash J ) \psi_{\overset{\sim}{\mathcal{B}}} = 0,\]
    where \( J \in \mathbf{k} \lBrack \hslash \rBrack\).

    In more detail, the quantisation of \(\mathcal{O}_W\) be given by \( ( \mathcal{O}_W \lBrack \hslash \rBrack , \star)\), where \( \star\) is the Moyal product, per definition (\ref{defn:star_prod_pois}) . Suppose \( \mathcal{O}_X \) is quantised to give \( (\mathcal{O}_X\lBrack \hslash \rBrack , \star_{\mathrm{Ext}})\), where \( \star_{\mathrm{Ext}}\) is the Moyal star product, defined on the extension \( 
   \mathcal{O}_X\lBrack \hslash \rBrack \simeq \mathcal{O}_W \lBrack \hslash \rBrack \otimes_{\mathbf{k}\lBrack \hslash \rBrack} \mathcal{O}_{G}^{G^{\perp}} \lBrack \hslash \rBrack \).  Then pick the Weyl-algebra representation of \( (\mathcal{O}_X\lBrack \hslash \rBrack , \star_{\mathrm{Ext}}) = \mathcal{W}_X\). Finally we study the formal completion \( \widehat{\mathcal{W}}_X \) along a maximal ideal. In  example, section (\ref{sec:examplereduct4d}), this will be at the point \(0 \in X\). The wavefunction is finally a generator of the cyclic module 
   \[ \langle \psi_{\overset{\sim}{\mathcal{B}}} \rangle =  \mathrm{Hom}\left( \widehat{\mathcal{W}}_X / \widehat{\mathcal{W}}_X \mathcal{J}_{\mathcal{B}}\lBrack \hslash \rBrack , \widehat{\mathcal{W}}_X \right).\]
   
    The second, defined by Kontsevich and Soibelman, is the formal Gaussian integral: 
    \[ \psi_{\widehat{\cB}} = \frac{1}{Z_0} \int_{W} \psi_G  \,  \psi_{\mathbb{L}}. \]
    In \cite[page 53]{ks_airy}, this is understood as an integral with \(\psi_G\) as a Gaussian measure and \(Z_0\) is some normalisation or constant. \( \psi_{G}\) is defined as the wavefunction supported on the image of \(G \) in the extension \(X\) of \(W\). Individually:
    \begin{align}
    \label{eqn:psigext}
        \langle \psi_{G} \rangle &=  \mathrm{Hom}\left( \widehat{\mathcal{W}}_X / \widehat{\mathcal{W}}_X \mathcal{J}_{G_X}\lBrack \hslash \rBrack , \widehat{\mathcal{W}}_X \right)\\
        \langle \psi_{\mathbb{L}} \rangle &=  \mathrm{Hom}\left( \widehat{\mathcal{W}}_X / \widehat{\mathcal{W}}_X \mathcal{J}_{\mathbb{L}}\lBrack \hslash \rBrack , \widehat{\mathcal{W}}_X \right).
    \end{align}
    where \(\mathcal{J}_{G_X}\lBrack \hslash \rBrack\) is defined below.

    \begin{rem}
    We use \( \psi_G\) to denote the wavefunction \(  \psi_{G_\Sigma} =\exp(Q_2(q,q')) \) as written in \cite[p. 53]{ks_airy}. We show that 
    \begin{equation} 
    \label{eqn:psiextg}
    \psi_G = P \psi_G,
    \end{equation} 
    where \(P\) is a function depending on the choice of extension, and \( \psi_G\) is a function satisfying the equations for \(G\) in \(W\). The function depending on the extension, \(P\), using the Kontsevich and Soibelman notation, contributes terms of \( \mathcal{O}( q q')\). We claim \(P\) defines a Fourier or Laplace like integral kernel. While \( \psi_G\) at most contributes \(\mathcal{O}(q^2)\) terms. So while \( \psi_G\) is quadratic in terms of the coordinates on the extension, we want to emphasise the important part is the \(P\). It is necessary that \(G\) is linear and coisotropic, or defined by a collection of linear equations. 
    \end{rem}

    We are effectively trying verify two propositions:
    \begin{prop}
    Symmplectic reduction commutes with deformation quantisation.
    \end{prop}
    Further:
    \begin{prop}
    \[ \psi_{\overset{\sim}{\mathcal{B}\,}} = \psi_{\widehat{\cB}} =\frac{1}{Z_0} \int_{V} \psi_G \,  \psi_{\mathbb{L}}\]
    \end{prop}
    We verify this explicitly for linear Lagrangians in section (\ref{sec:examplereduct4d}).

    One way to define the integral is to treat 
    \[ \frac{1}{Z_0} \int_{V} \psi_G\]
    as a linear functional and define the action on coordinates and then extend to arbitrary functions via Wick's theorem. 
    
    Geometrically the integral can be thought of as a convolution. We want to pick out the components of the function \( \psi_{\mathbb{L}}\) where \(G\) and \( \mathbb{L}\) intersect, to get a function \( \psi_{\widehat{\mathcal{B}}}\). A more basic example of this situation is from linear algebra. Given a pair of linear operators \(L_1\), \(L_2\), we want to find an element in the intersection of the kernels of \(L_1\) and \(L_2\). If the intersection is non trivial, and given \( v \in \mathrm{ker}(L_1)\), and \( w \in \mathrm{ker}(L_2)\), the projection of \(v\) onto \(w\) would lie in the intersection. We then want to project or reduce this down onto a smaller space.
    
    
    Explicitly, the ideal sheaf \(\mathcal{J}_{G_X} \lBrack \hslash \rBrack \) defining \( \psi_G \), is found from the extension \(X \cong W \oplus G/G^{\perp}\) of \(W\), where \(G\) is embedded as linear Lagrangian in \(X\). Note \(X\) has the opposite symplectic form. Recall in finite dimensions \(G\) is defined by \(n-g\) equations, then for \(X\) to be a \(2\,g\) dimensional extension, requires an extra \(2
    \,g\) equations, so \(G\) is sent to a Lagrangian defined by \(n+g\) equations, or an ideal \( \mathcal{J}_{G_X}\). On the level of rings this given by a map
    \[ \mathcal{O}_{X} \rightarrow \mathcal{O}_W. \]
    For example, in the affine case, introduce the \(2g\) coordinates \(z_1, \dots z_g, w_1, \dots w_g\), and then look for linear functions so \(G\) in \(X\) can be represented as \((y,w) + Q \cdot (z,x)=0\). We go through a full example in section (\ref{sec:examplereduct4d}).

    \subsection{Wavefunctions on linear Lagrangian and coisotropic subspaces}
    
    Let \(W\) be an \(2n\)-dimensional symplectic space, with Darboux coordinates \((x_i,y_i)\).

    \subsubsection{Linear Lagrangian}
    \label{sec:linearlag}
    We review the idea that there is a unique wavefunction associated to a linear Lagrangian subvariety in the symplectic space \(W\). From the deformation quantisation perspective, there is a module, supported on a formal neighbourhood of a point like example (\ref{ex:the_conic}). 
    
    Consider the case where \(\mathbb{L}\) is a codimension \(n\) linear Lagrangian subvariety, which is written, as a collection of equations in vector form, as:
    \begin{equation}
        \label{eqn:linearlag}
        y + Q \cdot x = 0,
    \end{equation}
    where \(Q\) is a \((n,n)\)-dimensional symmetric tensor. More generally, given \(n\) equations defining a Lagrangian \( \mathbb{L}\), we solve them for \(y_i(x^{\sbt})\). For quadratic and higher order Lagrangians, this requires completing at a maximal ideal, so \( \mathbb{L}\) is locally written as a generating function \(y_i dx_i =  dS_0(x^{\sbt})\), or \(y\) can be written as a formal series in \(x\) at a point. But in the linear Lagrangian case, we already have a form for \(y\), so \(S_0\) is found directly from equation (\ref{eqn:linearlag}): 
    \[S_0 = -\frac{1}{2} x \cdot Q \cdot x. \]
    If the equations are not full rank for \(y\), it is sufficient to set \(y_i=0\) at the point defined by the maximal ideal. It will also be possible to eliminate the corresponding \(x_i\) variables.
    
    Sparing some details about star-products, quantisation is representable by the Weyl-algebra, similarly to example (\ref{ex:the_conic}): 
    \[ x_i \rightarrow x_i, y_i \rightarrow \hslash \partial_i = \frac{\partial}{\partial x_i}.\]
    Now, a wavefunction \( \psi(x^{\sbt})\), is a formal series, in defined on a formal neighbourhood, that is the annihilator of the quantisation of the equations defining the Lagrangian subvariety. In the linear case it is equivalent to consider the quantisation of (\ref{eqn:linearlag}):
    \[ \hslash \frac{\partial}{\partial x} \psi + Q \cdot x\, \psi = 0. \]
    In the linear case any annihilator \( \psi\) is represented by the formal expression:
    \[ \psi = \mathrm{const} \, \exp \left( \frac{1}{\hslash} S_0 \right) = \mathrm{const} \, \exp \left( -\frac{1}{2\hslash} x \cdot Q \cdot x \right), \]
    where \( \mathrm{const}\) is an arbitrary constant.
    Note \(\exp(f) \) is a formal object which is a way to store information the module homomorphism from section (\ref{sec:wavefunctions}). Functions can be extracted from \(\exp\) via actions of \(\hslash \frac{\partial}{\partial x_i }\). Also \( \mathrm{const}\) is some constant. Note \(  d (S_0 ) = -(Q \cdot x) dx\), so \( y=\mathrm{graph}(dS_0)\). 
    
    More generally consider if the equations for \( \mathbb{L}\) can be written as
    \[ M \cdot y + N \cdot x = 0\]
    where either only one of \(M\) or \(N\) is not full rank. These equations must Poisson commute, as \(\mathbb{L}\) is a linear Lagrangian:
    \[ \{ M_{ij} y_j + N_{ij}x_j , M_{ks}y_s + N_{ks}x_s \} = 0.\]
    This gives the constraint \(- M_{ij} N_{kj}+N_{ij} M_{kj} = 0\). For \( \mathbb{L}\) to be Lagrangian, only one of \(M\) or \(N\) is allowed to not be full rank. 
    
    Consider the case where \(N\) is not full rank, then \(M\) is still invertible so
    \[ y + (M^{-1} N) \cdot x = 0,\]
    and the previous analysis follows where \( Q = (M^{-1} N)\). 
    
    Now consider the cases where \(M\) is not full rank. First consider the case where a single \(y_i\) does not appear in any equations. So \(M\) has a zero column and row. Multiple missing \(y_i\) will follow inductively. As \(N\) still has full rank, it is still possible to solve for the corresponding \(x_i\), so \(x_i = f_i(x_1, \dots x_{i-1}, x_{i+1}, \dots x_n)\), where \(f\) is a linear function. We then consider a reduced system of equations:
    \begin{equation} \label{eqn:reduced} \widetilde{M}\cdot \widetilde{y} + \widetilde{N} \cdot \widetilde{x} = 0,\end{equation}
    where \(\widetilde{M}, \widetilde{N} \) are the reduced tensors so there are no \(x_i\) or \(y_i\) terms. Now the the previous analysis follows.  While quantisation of \(W\) still includes \( \hslash \partial_i \), we only need to consider annihilators of the quantisation of equation (\ref{eqn:reduced}) in \( \mathbf{k} \lBrack x_1, \dots x_{i-1}, x_{i+1}, \dots x_n\rBrack\).

    \subsection{A formula for Gaussian integrals}

    Let \(W\) a symplectic space with coordinates \( (x^i,y_i)\), and \(V\) be a Lagrangian with coordinates \(x_i\). Recall the objective is to compare the wavefunctions, in the Weyl-algebra representation, first from the scheme theoretic quotient, the second from the Gaussian integral:
    \begin{equation} 
    \label{eq:gaussint}
    \psi_{\widehat{\mathcal{B}}} = \frac{1}{Z_0} \int_{V} \,  \psi_{\mathbb{L}} \, \psi_G  . 
    \end{equation}
    Note that \( \psi_{\mathbb{L}} \) depends on \(x\) coordinates, while \( \psi_G\) will depend on \(x\) and the coordinates of the extension. So the integral can just be viewed as integrating over \(V\).


    The \emph{central identity of quantum field theory} \cite{zee}, is a useful formula for evaluating Gaussian integrals, such as the integral for symplectic reduction of the wavefunction in equation (\ref{eq:gaussint}):
    \begin{lem}[Central identity of quantum field theory]
        \begin{equation} 
        \label{eqn:centralid}
        \frac{1}{Z_0}\int Dx \exp \left( \mathcal{Q}(x, J) \right) 
        = \exp\left(-V\left(\frac{\partial}{\partial J}\right)\right) \exp\left( \frac12 \, J \cdot (K')^{-1} \cdot J\right), 
    \end{equation}
    where 
    \[ \mathcal{Q}(x,J) = -\frac{1}{2} x \cdot K' \cdot x - V + J \cdot x.\]
    \end{lem} 
    Note that 
    \[\exp\left(-V\left(\frac{\partial}{\partial J}\right)\right),\] 
    is an operator, computed by expanding the exponential as a series. Also note these integrals are all treated formally. 
    \begin{rem}
    In the central identity, equation (\ref{eqn:centralid}), \(J\) is called a \emph{source} term.
    \end{rem}
    
    Now when \( \psi_G\) is given per equation (\ref{eqn:psiextg}), (for simplicity set \(S=0\), in general it will always factor out of the integral, and rescale by \(K\) and \(J\) by \(\hslash\)):
    \[ \psi_G = \exp\left( -\frac{1}{2}  x \cdot K \cdot x + J \cdot x \right),\]
    equation (\ref{eqn:centralid}) is useful for computing \( \psi_{\widehat{\mathcal{B}}}\). \(J\), from the extension of \(G\) in \(X\), is treated as the source term, so:
    \[ \psi_{\widehat{\mathcal{B}}} = \frac{1}{Z_0} \int Dx \exp\left( -\frac{1}{2}  x\cdot K \cdot x + J \cdot x \right) \psi_{\mathbb{L}}. \]

    
    \begin{ex} 
    Let \( \cL\) be determined by a collection of quadratic functions: 
    \[ I(\cL) = \langle -y_i + a_{ijk} x_j x_k + 2 b_{ijk} x_j y_k + c_{ijk} y_j y_k \rangle, \] 
    in \( \mathbf{k}[x_{\sbt},y_{\sbt}]\). Then there is a wavefunction of the form \(\psi_{\cL} = \exp(S)\), where \(S= \sum_{g\geq 0} h^{g-1} S_g\). \(S\) contains quadratic and linear terms in \(x\), so the formal integral for \( \psi_{\widehat{\mathcal{B}}}\) becomes:
    \begin{align*}
    \psi_{\widehat{\mathcal{B}}} =
    &\frac{1}{Z_0}\int Dx \exp \bigg[
    -\frac{1}{2} x \cdot \left(K - \sum_{g>0} \hslash^{g-1} S_{g,2;\bullet,\bullet} \right) \cdot x \\  
    & -V' + \left(J-\sum_{g>0} \hslash^{g-1} S_{g,1;\bullet}\right) \cdot x \bigg]
    \end{align*}
    where \(V'\) are the remaining terms of \(S\) cubic and greater order in \(x\).
    
    Using the central identity of quantum field theory, (\ref{eqn:centralid}) let \(J\) be a source term giving a formula for \( \psi_{\widehat{\mathcal{B}}}\):
    \begin{align*} 
    \psi_{\widehat{\mathcal{B}}}  =&  \exp\left(V'\left(\frac{\delta}{\delta J} \right) \right) \bigg[ \\ 
    & \exp\left( \frac{1}{2}\, (J-\sum_{g>0} \hslash^{g-1} S_{g,1;\bullet}) \cdot (K-\sum_{g>0} \hslash^{g-1} S_{g,2;\bullet,\bullet})^{-1}\cdot (J-\sum_{g>0} \hslash^{g-1} S_{g,1;\bullet}) \right) \bigg].  
    \end{align*}
    In general this could be impractical to evaluate explicitly.
    \end{ex}

\section{Deformation quantisation and reduction in 4-dimensions}
    \label{sec:examplereduct4d}
    
    We now give an explicit example of the reduction of a wavefunction, and show the integral formula computes the same result as a deformation quantisation of the symplectic reduction of schemes in the affine case. Although this is only a \(4\)-dimensional example, the computation would extend easily for arbitrary linear Lagrangians. 
    
    \begin{ex}    
    Consider the symplectic space \(W\), with coordinate ring 
    \[ \mathcal{O}(W) =   \mathbf{k}[x_1,x_2,y_1,y_2].\] 
    There is a Poisson bracket given by 
    \[\{x_i,x_j\}=\{y_i,y_j\}=0, \; \{x_i,y_i\}=1 .\] 
    Let \( \cL \subset W \) be a linear Lagrangian defined by the ideal \(I = \langle H_1, H_2\rangle\), where:
    \begin{align*}
        H_1 = y_1 + A x_1 +  B x_2,    \\
        H_2 = y_2 +  C x_1 + D x_2. 
    \end{align*}
    Recall \( I\) is required to be closed under the Poisson bracket, so \(B=C\).
    
    Let \(G\) be a linear coisotropic subspace of codimension \(n-g=1\). \(G\) is determined by the equation:
    \[ I(G) = \langle H_G :=a x_1 + b x_2 + c y_1 + d y_2 \rangle.\]
    \end{ex}

    Further \(G\) and \( \mathbb{L}\) must intersect transversally at the origin, so \( \mathcal{T}_0 G+ \mathcal{T}_0{\mathbb{L}} \cong \mathcal{T}_0W\). In this case, transversality is given by requiring at least one of following terms to be non zero:
    \begin{align}
    \label{eqn:transcons2dmin}
    a - cA - dB, \\ 
    b - cB-dD.
    \end{align}
    If both are zero, then \(G\) and \( \mathbb{L}\) do not intersect transversally. Equation (\ref{eqn:transcons2dmin}) can be found by row reducing the coefficient of the equations defining \(G\) and \( \mathbb{L}\) together. Alternatively equation (\ref{eqn:transcons2dmin}) can be found by requiring the wedge product of the differentials of the equations for \(G\) and \( \mathbb{L}\) being non zero. 

    Now note as \(G\) is codimension \(n-g=1\), the extension \(X = W\oplus G/G^{\perp}\) necessarily must have dimension \(6\). Let \( \mathcal{O}(X) = \mathbf{k}[x_1,x_2,z_1, y_1,y_2,w_1]\). Define a Poisson bracket  on \(\mathcal{O}(X)\) via \( \{x_i,y_i\}=\{z_1,w_1\} = 1\). Then on the algebraic level, the extension is defined by an opposite map on rings:
    \begin{equation} 
    \label{eqn:extgmaprings}
    \mathbf{k}[x_1,x_2,z_1,y_1,y_2,w_1] \rightarrow \mathbf{k}[x_1,x_2,y_1,y_2],
    \end{equation}
    such that the image of \(G\) in \(X\), \(G_X\) is a linear Lagrangian (linear is the simplest Lagrangian, and later quantisation will yield a Gaussian function), which in terms of rings means there is an ideal \(I(G_X)\) in \(\mathcal{O}(X)\) that is involutive under the Poisson bracket. 
    \[ \{I(G_X),I(G_X)\} = I(G_X).\]
    The map (\ref{eqn:extgmaprings}) is constructed on the generators. First 
    \begin{align*}
    x_i & \rightarrow x_i  \\
    y_i  & \rightarrow y_i.
    \end{align*}
    Then finding \(G_X\) as a Lagrangian in \(X\), requires finding two additional linear functions \(\zeta\), and \(\xi\):
    \begin{align}
    \label{eqn:solveforext}
    z_1 &= \zeta_1(x_{\sbt},y_{\sbt}), \\
    w_1 &= \xi_1(x_{\sbt},y_{\sbt}),
    \end{align}
    such that taking the equations \( z_1 - \zeta_1 = 0, w_1 - \xi_1 = 0\), and \(H_G = 0\) defines a Lagrangian \(G_X\) in \(X\). This means 
    \[ I(G_X)= \langle z_1 - \zeta_1,  w_1 - \xi_1, H_G \rangle, \]
    is a Poisson ideal in \(\mathcal{O}(X)\).  Using the Poisson bracket on \(X\), as they are all linear, we require 
    \[ \{H_G,z_1 -\zeta_1\} = \{ H_G,w_1 - \xi_1\} =   \{z_1 - \zeta_1,w_1 - \xi_1\} = 0.\]
    One choice is 
    \[\zeta_1 = a x_1 + c y_1, \quad  \text{and} \quad  \xi_1 = \frac{1}{cd} \left( d x_1 - c x_2 \right) .\]
    Now we can quantise \( \mathcal{O}(X)\), (and \( \mathcal{O}(W)\)) to find \( \psi_G\) and \( \psi_{\mathbb{L}}\), supported on formal completions \( \widehat{\mathbb{L}}\) and \( \widehat{G}_{\mathrm{Lag}}\) at \(0\). With this choice of \(G_X\) in \(X\), we construct a unique solution for \(\psi_G\).
    Consider the deformation quantisation of \( \mathcal{O}(X)\) with a star product, per definition (\ref{defn:star_prod_pois}). The non-commutative algebra \((\mathcal{O}(X)\lBrack \hslash \rBrack, \star)\) is isomorphic to a Weyl-algebra of operators with the map \[y_i \rightarrow \hslash \frac{\partial}{\partial x_i},  w_1 \rightarrow \hslash \frac{\partial}{\partial z_1}.\] 
    Then \( \psi_G(x_1,x_2,z_1)\) is a generator of a module, given by equation (\ref{eqn:psiextg}), and on a formal neighbourhood of \(x=0\), satisfies the equations: 
    \begin{align*}
       (a x_1 + b x_2) \psi_G + \hslash \left( c \frac{\partial}{\partial x_1 } +  d \frac{\partial}{\partial x_2} \right) \psi_G & = 0, \\
       ( z_1 - a x_1 ) \psi_G  - c \hslash \frac{\partial}{\partial x_1} \psi_G &= 0, \\ 
       \hslash \frac{\partial}{\partial z_1} \psi_G - \frac{1}{c d} ( d x_1 - c x_2) \psi_G &= 0.
    \end{align*}
    Then looking for a solution of the form:
    \[ \psi_G (x_1,x_2,z_1)= \mathrm{const} \, \exp\left( -\frac{1}{2} x \cdot K \cdot x  + J  \cdot x\right),\]
    on the completion around \(x=0\), we find: 
    \begin{align*}
        K = \frac{1}{\hslash} \left(\begin{array}{cc}
            a/c & 0 \\
            0 & b/d
        \end{array}\right), \quad J = \frac{1}{\hslash} \,  z_1 \left( \begin{array}{c}
            1/c \\
            -1/d
        \end{array}\right).
    \end{align*}
    where \( c , d \neq 0\). Note \( \mathrm{const}\) is just a generic constant term. Wavefunctions are elements of a module over \(\mathbf{k} \lBrack \hslash \rBrack\). Note unlike example (\ref{ex:the_conic}), there is no need to complete at \(y=0\) and pick a particular solution.
    

    Similarly in \(W\), there is a unique solution for \( \psi_{\cL}\):
    \begin{align*}
        \psi_{\cL}(x_1,x_2) =  \mathrm{const} \, \exp \left( -\frac{1}{2}   x \cdot Q \cdot  x \right),
    \end{align*}
    where 
    \[ Q = - \frac{1}{\hslash} \left( \begin{array}{cc}
        A  &  B \\
        B & D  
    \end{array}\right),\]
    and \( \mathrm{const}\) is another arbitrary constant in \( \mathbf{k}\lBrack \hslash \rBrack\).
    
    Now to finally perform the integration, \( \psi_{\cL}\) is a Gaussian, and via the central identity, equation (\ref{eqn:centralid}):
    \begin{align*} \psi_{\widehat{\mathcal{B}}}(z_1) &= \mathrm{const} \, \int D x \, \psi_{G} \,  \psi_{\mathbb{L}},  \\ 
    &= \mathrm{const} \, \int Dx \exp \left( -\frac{1}{2} x \cdot (K+Q) \cdot x + J \cdot x \right), 
    \end{align*}
    which gives
    \begin{equation*} \psi_{\widehat{\mathcal{B}}}(z_1) = \mathrm{const}\, \exp \left( \frac12\, J\cdot (K+Q)^{-1} \cdot J \right) = \mathrm{const} \exp\left( -\frac{1}{2 \hslash}\, c_0\, z_1^2\right),
    \end{equation*}
    where 
    \[ c_0 = \frac{\left(a c-A c^2+d (b-2 B c-d D)\right)}{ c d  \left(a (b-d D)-c \left(A (b-d D)+B^2 d\right)\right)}.\]
    Note that this integral fails to exist when \(Q+K\) does not have an inverse.  Checking for where \( \det(Q+K) = 0\), this is equivalent to when
    \[ B^2 c d = (a -  c A) ( b - d D), \]
    which is equivalent to the fact \( \mathbb{L}\) and \(G\) must intersect transversally. 

    %
    As a first check on this formula, taking the equations for \(G\), \(\cL\) in \(X\) and eliminating \(x_{\sbt}\) and \( y_{\sbt}\), gives an equation
    \begin{equation}
        \label{eqn:check1}
        w_1 = c_1 z_1,
    \end{equation}
    which algebraically, represents the sum of the ideals of \(G\) and \( \mathbb{L}\), per sequence (\ref{eqn:seqforB}). The sum of the ideals is the algebraic way to describe the intersection \( G \cap \mathbb{L}\). Also note
    \[ c_1 = \frac{   \left(-a c+A c^2+d (-b+2 B c+d D)\right)}{c d \left(d D (a-A c)-a b+A b c+B^2 c d\right)}. \]
    Note that \(c_1 = -c_0\).  As before, the quantisation \(\mathcal{O}(X)\lBrack \hslash \rBrack\), is isomorphic to a Weyl-algebra of operators:
    \[ x_i \rightarrow x_i, \, z_1 \rightarrow z_1, \, y_i \rightarrow \hslash \frac{\partial}{\partial x_i}, \, w_1 \rightarrow \hslash \frac{\partial}{\partial z_1}. \] 
    Quantising equation (\ref{eqn:check1}), gives a differential equation, 
    \[ \hslash \frac{\partial}{\partial z_1} \psi = c_1 z_1 \psi. \]
    This equation has a solution of the form:
    \[ \psi =  \mathrm{const} \, \exp\left( \frac{1}{2 \hslash} c_1 z_1^2 \right).\]
    Now as \(c_1 = -c_0\), \( \psi_{\widehat{\mathcal{B}}} = \psi\), so quantising on the quotient agrees with the integral formula.
    Both wavefunctions \( \psi\) and \(\psi_{\widehat{\mathcal{B}}}\) are undefined in the case where \(G\) and \( \mathbb{L}\) fail to intersect transversally.
    
    As an example, set \( \{ a,b,c,d\} \rightarrow 1\). In this case both \(\psi_{\widehat{\mathcal{B}}}\), and \( \psi\) reduce to
    \[ \psi = \mathrm{const} \exp \left( -\frac{1}{2 \hslash } \frac{ (2-A-2 B-D)}{ \left(B^2-(1-A)(1-D)\right)} z_1^2 \right),\]
    as long as \( B^2 \neq (1-A)(1-D)\), which is the transversality requirement.  So the quantisation of the reduction can agree with the integral transform from \cite{ks_airy}. However we did pick a particular deformation quantisation module for this to occur. In the case of the conic, example (\ref{ex:the_conic}), there was an ambiguity up to a constant term \( \hslash J\). In this example we chose particular constants for the correspondence between the integral formula and the reduction to work out cleanly.
    

%% file: ack.tex
\section*{Acknowledgements}
\addcontentsline{toc}{section}{Acknowledgements}

We would like to thank \href{https://researchers.ms.unimelb.edu.au/~norbury@unimelb/}{Paul Norbury} for valuable discussions. We also would like to acknowledge \href{https://www.matrix-inst.org.au/}{Matrix} and the organisers for the Quantum Curves, Integrability and Cluster Algebras workshop where some aspects of this article were discussed.